\documentclass[11pt]{article}
\usepackage[latin9]{inputenc}
\usepackage{verbatim}
\usepackage{hyperref}
\usepackage{amsmath}
\usepackage{amsthm}
\usepackage{amssymb}

\makeatletter
\usepackage{amsthm}
\usepackage{tikz}
\usepackage{bbm}
\usepackage{a4wide}
\usepackage{cleveref}

\usepackage{times}

\newtheorem{theorem}{Theorem}
\newtheorem{lemma}[theorem]{Lemma}
\newtheorem{definition}[theorem]{Definition}

\newcommand{\LB}{\mathrm{LB}}
\newcommand{\OPT}{\mathrm{OPT}}
\newcommand{\cost}{\mathrm{cost}}
\newcommand{\len}{\mathrm{len}}
\newcommand{\leftjobs}{\overleftarrow{J}}
\newcommand{\rightjobs}{\overrightarrow{J}}
\newcommand{\rightb}{\overrightarrow{b}}
\newcommand{\hinten}{\mathrm{end}}
\newcommand{\vorne}{\mathrm{beg}}

\usepackage[colorinlistoftodos,shadow,textsize=tiny,textwidth=2cm,color=green!50!gray]{todonotes} %,disable
\makeatother

\newcommand{\email}[1]{\href{mailto:#1}{#1}}

\begin{document}
\global\long\def\N{\mathbb{N}}%
\global\long\def\I{\mathcal{I}}%

\title{Simpler constant factor approximation algorithms for weighted flow time -- now for any $p$-norm\footnote{We thank Schloss Dagstuhl for hosting the Seminar 23061 on Scheduling in February 2023 where we had fruitful discussions on this topic.}}
\date{}
\author{Alexander Armbruster\footnote{TU of Munich, \email{alexander.armbruster@tum.de}} \and Lars Rohwedder~\footnote{Maastricht University, \email{l.rohwedder@maastrichtuniversity.nl}, supported by Dutch Research Council
(NWO) project ``The Twilight Zone of Efficiency: Optimality of Quasi-Polynomial Time Algorithms'' [grant number
OCEN.W.21.268]} \and Andreas Wiese\footnote{TU of Munich, \email{andreas.wiese@tum.de}}}
\maketitle

\begin{abstract}
A prominent problem in scheduling theory is the weighted flow time
problem on one machine. We are given a machine and a set of jobs,
each of them characterized by a processing time, a release time,
and a weight. The goal is to find a (possibly preemptive) schedule
for the jobs in order to minimize the sum of the weighted flow times,
where the flow time of a job is the time between its release time
and its completion time. It had been a longstanding important open
question to find a polynomial time $O(1)$-approximation algorithm
for the problem. In a break-through result, Batra, Garg, and Kumar
(FOCS 2018) presented such an algorithm with pseudopolynomial running
time. Its running time was improved to polynomial time by Feige, Kulkarni,
and Li (SODA 2019). The approximation ratios of these algorithms are
relatively large, but they were improved to $2+\epsilon$ by Rohwedder
and Wiese (STOC 2022) and subsequently to $1+\epsilon$ by Armbruster,
Rohwedder, and Wiese (STOC 2023).

All these algorithms are quite complicated and involve for example
a reduction to (geometric) covering problems, dynamic programs to
solve those, and LP-rounding methods to reduce the running time to
a polynomial in the input size. In this paper, we present a much
simpler $(6+\epsilon)$-approximation algorithm for the problem that does
not use any of these reductions, but which works on the input jobs
directly. It even generalizes directly to an $O(1)$-approximation
algorithm for minimizing the $p$-norm of the jobs' flow times, for
any $0 < p < \infty$ (the original problem setting corresponds to $p=1$). Prior to our work, for $p>1$ only a pseudopolynomial time $O(1)$-approximation algorithm was known for this variant, and no algorithm for $p<1$.

For the same objective function, we present a very simple QPTAS for
the setting of constantly many unrelated machines for
$0 < p < \infty$ (and assuming quasi-polynomially
bounded input data). It works in the cases with and without the possibility
to migrate a job to a different machine. This is the first QPTAS for
the problem if migrations are allowed, and it is arguably simpler
than the known QPTAS for minimizing the weighted sum of the jobs'
flow times without migration.
\end{abstract}

\thispagestyle{empty}
\newpage
% this is on purpose, do not remove!
\setcounter{page}{1}

\section{Introduction}

Weighted flow time is an important problem in scheduling that has
been studied extensively in the literature, e.g.,~\cite{DBLP:journals/siamcomp/BansalP14,Batra0K18,feige2019polynomial,RohwedderW21,armbruster2023ptas}.
We are given one or multiple machines and a set of jobs $J$. Each
job $j\in J$ is characterized by a processing time $p_{j}\in\N$,
a release time $r_{j}\in\N_{0}$, and a weight $w_{j}\in\N$. The
goal is to compute a schedule for the given machines that respects
the release times. For each job $j\in J$, we denote by $C_{j}$ the
completion time of $j$ in a computed schedule, and we define its
\emph{flow time} by $F_{j}:=C_{j}-r_{j}$. Hence, the flow time of
a job is the time that it stays unfinished in the system. The objective
is to minimize the total weighted flow time $\sum_{j\in J}w_{j}F_{j}$
of the jobs. Note that this also minimizes the average weighted flow
time of the jobs. Weighted flow time is a natural objective function
since it is desirable to finish each job as quickly as possible after
it has been released. Throughout this paper, we will assume that we
can preempt the jobs, i.e., we can start working on a job $j\in J$
on some machine and then interrupt and resume it later.

Despite a lot of research on the problem, it had been open
for a long time to find a polynomial time $O(1)$-approximation
algorithm for the setting of a single machine. For example,
in an influential survey by Schuurman and Woeginger~\cite{schuurman-woeginger},
it was considered as one of the ten most important open problems in machine
scheduling.
In a breakthrough result,
Batra, Garg, and Kumar~\cite{Batra0K18} presented an $O(1)$-approximation
algorithm with pseudopolynomial running time for the problem. Subsequently,
Feige, Kulkarni, and Li gave a black-box reduction to transform such
an algorithm into a polynomial time algorithm (while losing only a
factor of $1+\varepsilon$)~\cite{feige2019polynomial}. The algorithm in~\cite{Batra0K18} has a very large
approximation ratio; however, this ratio was later improved to $2+\varepsilon$
by Rohwedder and Wiese~\cite{RohwedderW21} and finally to $1+\varepsilon$ by
Armbruster, Rohwedder, and Wiese~\cite{armbruster2023ptas}.

All these polynomial time $O(1)$-approximation algorithms first reduce
the problem to a (geometric) covering problem. Then, this covering
problem is solved via a pseudopolynomial time dynamic program (DP).
To obtain a polynomial time algorithm, the mentioned black-box reduction
in~\cite{feige2019polynomial} partitions a given instance into a linear number of
instances with polynomially bounded input data. Then, it combines
solutions for those instances to a global solution using LP-rounding
techniques.

Due to these transformations, the resulting algorithms are quite complicated.
Also, it is not very transparent what their steps concretely mean
for the given input jobs. Due to the reductions to the different covering
problems and their complicated DPs, this scheduling viewpoint is lost.
In particular, it would be desirable to have an algorithm that works
on the input jobs directly. Then, one could understand better how
this algorithm makes decisions for the given jobs and how properties
of the problem are used.

Since we have a PTAS for weighted flow time on one machine, a natural
next question is how we can solve the problem on two machines, or
maybe on a constant number of machines. For these cases, no polynomial
time $O(1)$-approximation algorithms are known. There is a quasi-polynomial
time $(1+\varepsilon)$-approximation algorithm if the jobs are not
allowed to be migrated~\cite{bansal2005minimizing}, i.e., if for each job we need to select
a machine on which we process it completely (possibly preemptively).
However, no such algorithm is known in the setting with migration.
Another natural generalization is to minimize the $p$-norms of the
weighted jobs' flow times, i.e., to minimize $(\sum_{j}w_{j}(F_{j})^{p})^{1/p}$
where $0<p<\infty$ (i.e., the case $p=1$ corresponds to minimizing
the sum of the weighted flow times). The mentioned pseudopolynomial
time $O(1)$-approximation algorithm~\cite{Batra0K18} extends to the case where
$p\ge1$; however, it is not clear whether this is also the case for
the black-box reduction in~\cite{feige2019polynomial}, and how to solve the case where
$p<1$.

\subsection{Our contribution}

In this paper, we present a $(6+\varepsilon)$-approximation algorithm
for weighted flow time on a single machine. Our algorithm is simpler
than the previous $O(1)$-approximation algorithms for the problem.
In particular, we do not use any of the reductions mentioned above,
but we compute a solution for the given jobs directly. This makes
it arguably easier to understand what the steps of our algorithm mean
for these jobs. Still, our approximation ratio is much smaller than
ratio of first $O(1)$-approximation algorithm~\cite{Batra0K18} which was at
least $10,\!000$. Also, our algorithm works directly in polynomial
time, and, therefore, it does not need the black-box reduction from~\cite{feige2019polynomial}.

Our algorithm is a recursion that we embed into a DP in order to bound
its running time by a polynomial. It partitions the overall problem
step by step into smaller subproblems. In each step, we make some
partial decisions about the completion times for some of the jobs.
The corresponding subproblems are carefully designed such that they
have only limited interaction between each other. On a high level,
each subproblem is characterized by a subinterval $[s,t)$ and a set
of jobs $J(s,t)$. We need to decide for each job $j\in J(s,t)$ whether
we want to complete it before time $s$, during $[s,t)$, or after
$t$. In case we decide that $j$ completes before time $s$ or after
time $t$, then in the current subproblem we do not select the final
completion time of $j$ yet; this will be selected in some other subproblem.
One key insight is that if a job $j$ is released long before $s$
and completes during $[s,t)$, then for its flow time it does not
matter much where exactly during $[s,t)$ it completes. Therefore,
for each such job $j$ we allow only two decisions in the current
subproblem: either we decide that it completes until time $s$ or
we allow that its completion time is $t$ or larger.
Intuitively, we guess how much time we want to work until time $s$
on such jobs in $J(s,t)$ (that are released long before $s$). Then,
we use the Lawler-Moore algorithm to select which of these jobs we
want to scheduling during that time. On the other hand, if a job $j\in J(s,t)$
is released shortly before $s$ and completes during $[s,t)$, then
for $j$ it \emph{does }matter when during $[s,t)$ we complete it.
Therefore, we recurse on such jobs in two different subproblems that
correspond to two subintervals of $[s,t)$.

First, we present a version of our algorithm that runs in pseudopolynomial
time. We cannot bound the running time by a polynomial because each
DP-cell has one entry for a point in time $b$ for which we cannot
directly bound the set of possible values by a polynomial. This is
similar to the classical DP for the knapsack problem in which each
cell is characterized by the remaining space in the knapsack (for
which there could be more than a polynomial number of options) and
the next item in the list. We transform our pseudopolynomial time
DP into a polynomial time algorithm with the same method that is used
in the mentioned DP for knapsack. We discretize and round the costs
of our solutions and we characterize each DP-cell by a bound for this cost
instead of $b$. Then, given such a cost bound, our DP computes a
solution with the best possible value for $b$. In this way, we circumvent
the black-box reduction in~\cite{feige2019polynomial} and directly obtain a polynomial
time algorithm. Our algorithm extends directly to the setting of minimizing
the $p$-norm of the weighted job's flow times, i.e., to minimize~$(\sum_{j}w_{j}(F_{j})^{p})^{1/p}$.
Our approximation ratio is a constant for any fixed $p>0$, it is bounded by $6+\epsilon$
for any $p\ge 1$, and it tends to $2+\epsilon$ for $p\rightarrow \infty$.
In particular, this yields the first \emph{polynomial} time $O(1)$-approximation
algorithm for this setting (recall that it is not clear whether the
black-box reduction in~\cite{feige2019polynomial} can be extended to it).

Our second result is a quasi-polynomial time $(1+\varepsilon)$-approximation
algorithm for minimizing~$(\sum_{j}w_{j}(F_{j})^{p})^{1/p}$ in the
setting of $O(1)$ machines, assuming that the job's processing times
are quasi-polynomially bounded integers. For this, we allow even that
the machines are unrelated, i.e., that for each combination of a job
$j$ and a machine $i$ we are given a processing time $p_{ij}\in\N\cup\{\infty\}$
that the job $j$ would need if we executed it completely on machine
$i$. We obtain our result in the settings with and without migration,
i.e., with and without the possibility to interrupt a job
on one machine and resume it
(later) on some other machine.

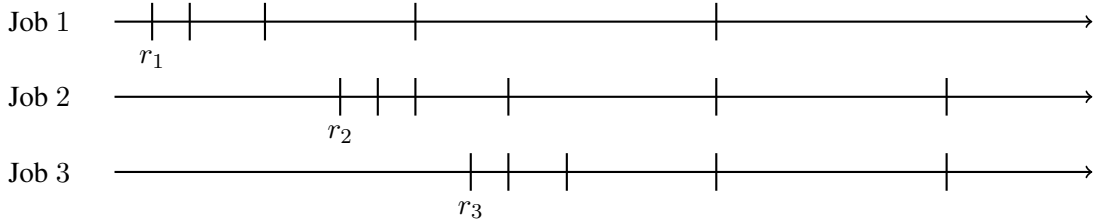
\begin{figure}
\begin{center}
\begin{tikzpicture}
		\node at (-4, 0) {Job $1$};
		\draw[->, thick] (-3,0) -- (10,0);
		
		\draw[-, thick] (-2.5,0.25) -- (-2.5,-0.25);
		\node[] at (-2.5, -0.5) {$r_1$};
		\draw[-, thick] (-2,0.25) -- (-2,-0.25);
		\draw[-, thick] (-1,0.25) -- (-1,-0.25);
		\draw[-, thick] (1,0.25) -- (1,-0.25);
		\draw[-, thick] (5,0.25) -- (5,-0.25);
		
		\node at (-4, -1) {Job $2$};
		\draw[->, thick] (-3,-1) -- (10,-1);
		\draw[-, thick] (0,-0.75) -- (0,-1.25);
		\node[] at (0, -1.5) {$r_2$};
		\draw[-, thick] (0.5,-0.75) -- (0.5,-1.25);
		\draw[-, thick] (1,-0.75) -- (1,-1.25);
		\draw[-, thick] (2.236,-0.75) -- (2.236,-1.25);
		\draw[-, thick] (5,-0.75) -- (5,-1.25);
		\draw[-, thick] (8.062,-0.75) -- (8.062,-1.25);

		\node at (-4, -2) {Job $3$};
		\draw[->, thick] (-3,-2) -- (10,-2);
		\draw[-, thick] (1.736,-1.75) -- (1.736,-2.25);
		\node[] at (1.736, -2.5) {$r_3$};
		\draw[-, thick] (2.236,-1.75) -- (2.236,-2.25);
		\draw[-, thick] (3.013,-1.75) -- (3.013,-2.25);
		\draw[-, thick] (5,-1.75) -- (5,-2.25);
		\draw[-, thick] (8.062,-1.75) -- (8.062,-2.25);
\end{tikzpicture}
\end{center}
\caption{Alignment of deadlines and intervals in QPTAS}
\label{fig:intervals}
\end{figure}

Our algorithm is surprisingly simple. We assume that the jobs are
ordered non-decreasingly by their release times. For each job $j$,
we introduce logarithmically many candidate deadlines. If we
assign a deadline $d_{j}$ to a job $j$, we decide to complete $j$
until time $d_{j}$ and pay for $j$ as if $j$ completed exactly
at time~$d_{j}$. We define these candidate deadlines densely enough
so that we lose at most a factor $1+\varepsilon$ if a job completes
strictly between two of them.
Also, this yields logarithmically many
intervals between any two consecutive deadlines for a job $j$. For
any two jobs $j,j'$, these intervals are aligned (see Figure~\ref{fig:intervals}),
i.e., any two such intervals are either disjoint or one is contained
in the other.

First, our algorithm guesses for how long we want to execute the first
job during each of its intervals. Then we recurse on the next job.
We embed the whole procedure into a DP in which we have a subproblem
for each combination of a (next) job $j$ and the information for
how long each subset of the machines is still available during each
interval for $j$. Given such a subproblem, we enumerate all possibilities
for how long we work on $j$ during each of its intervals and on which
machines. Each such possibility yields a DP-cell for the next job
$j'$. If the intervals for $j'$ are finer than the intervals for
$j$, we additionally enumerate all possibilities how the work during
each interval for $j$ is distributed over the intervals for $j'$.

Our QPTAS is arguably simpler than the known QPTASs for minimizing
the sum of the weighted jobs' flow times on one or $O(1)$ machines~\cite{chekuri2002approximation, bansal2005minimizing}.
Also, we remark that prior to our work, no \mbox{QPTAS} was known for minimizing
$\sum_{j\in J}w_{j}F_{j}$ on $O(1)$ machines in the setting with
migration, not even for~$O(1)$ identical machines.
We remark that in~\cite{bansal2005minimizing} it is explicitly mentioned
that it is not clear how to adapt the used methods to the setting with migration
(and the QPTAS in \cite{chekuri2002approximation} works only on one machine).

\subsection{Other related work}
Prior to the mentioned first constant factor approximation algorithm for weighted flow time on a single machine~\cite{Batra0K18}, a polynomial time $O(\log \log P)$-approximation for the general scheduling problem~(GSP) was presented by Bansal and Pruhs~\cite{DBLP:journals/siamcomp/BansalP14}. Here, $P$ denotes the ratio between the largest and the smallest processing time in a given instance. GSP is a generalization of weighted flow time on a single machine and it can model various different objective functions, including the $p$-norm of the weighted flow times of the jobs.
Prior to the PTAS for minimizing weighted flow time on a single machine~\cite{armbruster2023ptas}, PTASs were known for the special case where $w_j=1/p_j$ for each job $j\in J$~\cite{DBLP:journals/scheduling/BenderMR04, chekuri2002approximation} and for the case that~$P=O(1)$~\cite{chekuri2002approximation}.

Weighted flow time has also been studied in the online setting. There is an online algorithm with a competitive ratio of $O( \min \{\log W, \log P, \log D\})$ due to Azar and Touitou~\cite{azar2018improved}, where $W$ and $D$ are the ratios between the largest and smallest job weights and densities, respectively, where the density of a job $j$ is defined as $w_j/p_j$. This result improves and unifies the previously known $O(\log P)$-competitive algorithm by Chekuri, Khanna, and
Zhu~\cite{chekuri2001algorithms} and the $O(\log W)$-competitive algorithm by Bansal and Dhamdhere~\cite{bansal2007minimizing}. Previously, Bansal and Dhamdhere~\cite{bansal2007minimizing} presented also a semi-online $O(\log nP)$-competitive algorithm. Furthermore, there cannot be an online algorithm with a constant competitive ratio, as shown by Bansal and Chan~\cite{bansal2009weighted}.

In the setting with multiple machines there is the mentioned QPTAS by Bansal~\cite{bansal2005minimizing} which works only in the setting without job migrations (however, preemptions are allowed).
In addition, there is an FPT-$(1+\varepsilon)$-approximation algorithm known where the parameters are the number of machines, $\varepsilon$, and upper bounds on the (integral) jobs processing times and weights~\cite{wiese:LIPIcs:2018:9432}.

\section{Constant approximation algorithm for single machine\label{sec:algorithm} }
In this section, we present a 6-approximation algorithm for minimizing
sum of weighted flow times on a single machine with pseudopolynomial running
time. In Appendix~\ref{sec:polytime},
we will transform it into a $(6+\varepsilon)$-approximation
algorithm with polynomial running time,
and into a polynomial time $O(1)$-approximation algorithm
for any $p$-norm.

In our algorithm, we do not compute a schedule directly; instead,
for each job $j$ we compute a deadline $d_{j}$ such that $j$ needs
to be completed by time $d_{j}$. Given these deadlines, we obtain
a schedule by using the earliest-deadline-first (EDF) algorithm. This
algorithm starts at time $0$ and at each point in time, it works
on a job with the earliest deadline among all jobs that have been
released but not yet finished. Ties are broken in an arbitrary, but
globally consistent way. EDF is optimal in the sense that it finds
a feasible schedule (i.e., a schedule that respects all deadlines)
if such a schedule exists. This is shown in the following lemma
which also gives a characterization when the job deadlines allow
a feasible schedule.

\begin{lemma}[\cite{DBLP:journals/siamcomp/BansalP14}]\label{lem:EDF}
Given a set of jobs $J$ with processing times $p_{j}$, release dates
$r_{j}$ and deadlines $d_{j}$, the following statements are equivalent: 
\begin{itemize}
\item there exists a preemptive schedule in which each job $j$ finishes
before its deadline $d_{j}$,
\item EDF computes such a schedule,
\item for each interval $[s,t]$ with $\min_{j}r_{j}\leq s\leq t\leq\max_{j}d_{j}$
it holds that $\sum_{j\in J:s\leq r_{j}\leq d_{j}\leq t}p_{j}\leq t - s$.
\end{itemize}
\end{lemma} 

\subsection{Dynamic programming table}
Our algorithm is a dynamic program. The intuition for our DP-cells
is the following: each DP-cell corresponds to an interval $[s,t)$ and a starting time $b$. It concerns only some jobs $J(s, t)\subseteq J$ that we will specify later.
For each job $j\in J(s, t)$ we want to
\begin{itemize}
\item determine that its deadline $d_{j}$ is at most $s$, or
\item determine that its deadline $d_{j}$ is at least $t$, or
\item specify a deadline $d_{j}$ such that $d_{j}\in(s,t)$.
\end{itemize}
We have the restriction that before time $b$ we are
not allowed to work on jobs in $J(s, t)$. Intuitively, this time is reserved
for jobs in $J\setminus J(s, t)$. Given such a DP-cell, we will identify
a set of jobs $\overleftarrow{J}\subseteq J(s, t)$ which are released long before $s$. In particular, for such jobs it makes only a small difference
in the cost whether they finish at time $s$ or at time $t$. We guess
the total processing time of jobs in $\leftjobs$ that complete before
$s$ in the optimal solution to the DP-cell. Given this as a constraint,
we use the Lawlers-Moore algorithm (details below) 
to decide which jobs in $\leftjobs$ finish
before $s$ and which finish after $t$, minimizing the resulting
cost. Then, for defining the deadlines of the jobs in $J(s, t)\setminus \leftjobs$,
we make a recursive call to two DP-cells corresponding to the (smaller)
intervals $[s,(s+t)/2)$ and $[(s+t)/2,t)$.
Finally, we carefully merge the recursively computed schedules together. 

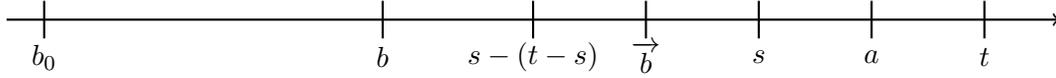
\begin{figure}
\begin{centering}
\begin{tikzpicture}
		\draw[->, thick] (-4,0) -- (10,0);
		\draw[-, thick] (1,0.25) -- (1,-0.25);
		\node[] at (1, -0.5) {$b$};
		
		\draw[-, thick] (-3.5,0.25) -- (-3.5,-0.25);
		\node[] at (-3.5, -0.5) {$b_0$};
		
		\draw[-, thick] (3,0.25) -- (3,-0.25);
		\node[] at (3, -0.5) {$s-(t-s)$};
		\draw[-, thick] (4.5,0.25) -- (4.5,-0.25);
		\node[] at (4.5, -0.5) {$\rightb$};

		\draw[-, thick] (6,0.25) -- (6,-0.25);
		\node[] at (6, -0.5) {$s$};

		\draw[-, thick] (7.5,0.25) -- (7.5,-0.25);
		\node[] at (7.5, -0.5) {$a$};
		
		\draw[-, thick] (9,0.25) -- (9,-0.25);
		\node[] at (9, -0.5) {$t$};
	\end{tikzpicture}
\par\end{centering}
\caption{\label{fig:DP-cells}An example for the values $s,t,b,$ and $b_{0}$
in the definition of a DP-cell $(s,t,b)$ and the value $\protect\rightb$
and $a$ that arise in its computation.}
\end{figure}

Formally, let $T$ be the smallest power of $2$ that is larger
than $\max_{j\in J} r_j + \sum_{j\in J} p_j$.
In particular, if a schedule does not
have unnecessary idle time, then it will complete all jobs by time $T$.
Thus, we may restrict our attention to the time interval $[0, T)$.
Each DP-cell is defined by a tuple $(s, t, b)\in \mathbb{Z}^{3}$
with $0\leq b\leq s<t\leq T$ (see Figure~\ref{fig:DP-cells}).
We define a set of intervals $\I$ resembling
the a binary tree structure and require for each DP-cell that
$[s,t)\in\I$. We define that $[0,T)\in\I$ which intuitively
forms the \emph{root} of the tree. Recursively,
we define that each interval $I=[s,t)\in\mathcal{I}$ with $t-s>1$
has a \emph{left child} $[s,(s+t)/2)\in\mathcal{I}$
and a \emph{right child} $[(s+t)/2,t)\in\mathcal{I}$.
Each interval $I=[s,t)\in\mathcal{I}$ with $t-s=1$ is a \emph{leaf}
and does not have any children. Given $[s,t)\in\I$, we define an
\emph{earliest starting time} $b_{0}$ as 
\begin{equation*}
b_{0}:=
\begin{cases}
\max\{0,s-2(t-s)\} &\text{if $(s,t)$ is a left child in $\mathcal{I}$,
i.e. $[s,t+(t-s))\in\mathcal{I}$ and}\\
\max\{0,s-3(t-s)\} &\text{if $(s,t)$ is a right child in $\mathcal{I}$,
i.e. $[s-(t-s),t)\in\mathcal{I}$}.
\end{cases}
\end{equation*}
For the root interval $[0,T)$ we define $b_{0}:=0$. We require for
each DP-cell that $b\geq b_{0}$, i.e., otherwise we do not introduce
a DP-cell for a tuple $(s,t,b)$.

The subproblem corresponding to a DP-cell $(s,t,b)$ is the following:
Let $J(s, t)$ be the set of all jobs $j\in J$ with $b_{0}\leq r_{j} < t$.
A solution to this subproblem consists of a deadline $d_{j}\in[s,t)\cup\{\infty\}$
for each job $j\in J(s, t)$, such that there exists a schedule processing
all jobs in $J(s, t)$ after $b$ and before their respective deadlines.
In order to define the cost of such a solution, for each job $j\in J(s, t)$
we define a function $\cost_{j}$ as 
\[
\cost_{j}(d)=\begin{cases}
0 & \text{ if }d=s\\
w_{j}\cdot(\min\{d,t\}-r_{j}) & \text{ if }d>s
\end{cases}
\]
 and define the cost of the computed solution as $\sum_{j\in J(s, t)}\cost_{j}(d_{j})$. 

Intuitively, if for some job $j\in J(s, t)$ we define $d_{j}=s$, then
in this DP-cell we determine only that $j$ completes at some point before
$s$ (possibly much earlier than $s$) and the precise deadline of
$j$ is determined by some other DP-cell. If $d_{j}\in (s,t)$ then in
this DP-cell we determine the precise deadline $d_{j}$ of $j$ and pay
$w_{j}(d_{j}-r_{j})$, i.e., the weighted flow time of $j$. If $d_{j}=\infty$,
then in this DP-cell we determine only that $d_{j}\ge t$, the precise deadline
of $j$ is determined in some other DP-cell, and in this DP-cell we pay
only for $w_{j}\cdot(t-r_{j})$, i.e., the weighted flow time of $j$
up to time $t$.

The initially given problem instance corresponds to the subproblem
for the DP-cell $(0,T,0)$ (note that then $J(s, t)=J$). We will show
that all jobs finish by time $T$ if we run EDF with the computed deadlines
(see in Lemma~\ref{lem:makespan}). Hence, also a job $j$ with $d_{j}=\infty$
completes by time $T$, and thus we could equivalently define $d_{j}=T$
instead for all such jobs.

\subsection{Computing values of DP-cells}
Before we describe how we compute a solution to each DP-cell, we review
the Lawler-Moore-algorithm, which we will use as a subroutine for this
task. We use it to solve the following subproblem: Given is a set
of jobs $\widehat{J}$ where each job $j\in\widehat{J}$ is characterized
by a processing times $\widehat{p}_j\in\N$, a release time $\widehat{r}_j\in\N_{0}$,
and a penalty cost $\widehat{c}_j\in\N_{0}$. In addition, there is
a common deadline $\widehat{d}$ for all jobs. The goal is to compute
a possibly preemptive schedule for the jobs in $\widehat{J}$. In such
a schedule, let $S\subseteq\widehat{J}$ be all jobs
that complete strictly after than $\widehat{d}$. Then its cost 
is $\sum_{j\in S}\widehat{c}_j$.
\begin{lemma}[Lawler-Moore algorithm]
\label{lem:Lawler} Given a set of $\widehat n$ jobs $\widehat{J}$ with release dates
$\widehat{r}_j$, processing time $\widehat{p}_j$, cost $\widehat{c}_j$
and a common deadline $\widehat{d}$. In time $O(\widehat n\sum_{j}\widehat{p}_j)$ we can
compute a schedule with minimum cost for the above problem. 
\end{lemma}
We process the DP-cells $(s,t,b)$ in non-decreasing order of the lengths
of the intervals $t-s$. 
Let~$\leftjobs$ denote the set
of jobs $j\in J(s,t)$ with $r_{j}\le s-(t-s)$ and define $\rightjobs := J(s,t) \setminus \leftjobs$.

We consider first $\leftjobs$.
For each job $j\in\leftjobs$, we
have that $\cost_{j}(\infty)=w_{j}(t-r_{j})\le 2\cdot\cost_{j}(t')$ for
any $t'\in (s,t]$. Hence, up to a factor of 2, we are indifferent whether
we assign $j$ a deadline in $(s,t]$ or a deadline of $\infty$.
Therefore, for each job $j\in \leftjobs$, we allow only the deadlines
$d_{j} = s$ and~$d_{j} = \infty$.
Consider the schedule corresponding to the optimal solution to $(s,t,b)$.
We define $\rightb$ as the amount of time during $[0, s]$ where the machine
does \emph{not} process jobs in $\rightjobs$. It is not hard to
see (details in the analysis) that one may assume without loss of generality that
all jobs in $\leftjobs$ with deadline $s$ are only processed before time
$\rightb$ and $\rightjobs$ is only processed after time $\rightb$.
We guess $\rightb$, i.e., we try each value
$\rightb\in \mathbb N_0$ with $\max\{b,s-(t-s)\}\leq \rightb\leq s$.

Using Lemma~\ref{lem:Lawler}, we then compute a minimal cost schedule
to the jobs in $\widehat{J}:=\leftjobs$, where for each job $j\in\widehat{J}$
we define $\widehat{r}_j:=\max\{r_{j}, b\}$, $\widehat{c}_j:=\cost_{j}(\infty)$,
and $\widehat{d}=\rightb$. Consider the resulting schedule for the jobs
$\widehat{J}$. We extract from this schedule the information which
jobs in $\widehat{J}$ finish before $\widehat{d}$ and which jobs finish
after $\widehat{d}$. For each job $j\in\widehat{J}$ finishing at time $\widehat{d}$
or earlier, we define $d_{j}:=s$; for each job $j\in\widehat{J}$ finishing
after time $\widehat{d}$, we define $d_{j}:=\infty$. 

It remains to define deadlines for $\rightjobs$, which
contains all jobs released in $(s-(t-s),t)$.
If $t - s = 1$ then these are exactly the jobs released at time $s$
and therefore none of them can complete before $t$. In this case,
we set $d_j = \infty$ for each job $j\in \rightjobs$.
Now assume that $t - s > 1$.
To define the deadlines for the jobs in $\rightjobs$,
we use the precomputed deadlines from two other DP-cells, which correspond to child intervals
of $[s,t)$ in $\I$.

Let $a=(s+t)/2$ and for the solutions in DP-cells $(s,a,\rightb)$ and $(a,t,\rightb)$ denote the sets of deadlines by $\{ d_{j}^{(1)}\} _{j\in J(s,a)}$ and
$\{ d_{j}^{(2)}\} _{j\in J(a,t)}$, respectively. Notice that $J(a,t) = \rightjobs$.
If for some job $j$ the right subproblem assigned it a deadline of $a$, i.e., $d_{j}^{(2)}=a$,
then also $j\in J(s,a)$ and,
intuitively, job $j$ should finish by time $a$ the latest. In that case, we define $d_j$ according to the left subproblem, i.e.,
$d_j:=d_{j}^{(1)}$. Otherwise, define $d_j$ according to the right subproblem, i.e., $d_j:=d_{j}^{(2)}$. Formally,
\begin{equation*}
d_{j}:=\begin{cases}
d_{j}^{(2)} & \text{ if }d_{j}^{(2)}>a \\
\min\{d_{j}^{(1)}, a \} & \text{ otherwise.}
\end{cases}
\end{equation*}
Thus, we compute a deadline for every job in $J(s, t)$, assuming a fixed
guess for $\rightb$. We store in the our
DP-cell $(s,t,b)$ the deadlines corresponding to the smallest
total cost $\sum_{j\in J(s,t)}\cost_{j}(d_{j})$ over all enumerated values
of $\rightb$.

Finally we output the EDF-schedule according to
the deadlines in the cell $(0,T,0)$. Since each job
finishes by time $T$ the latest (which we will prove in Lemma~\ref{lem:makespan}),
for each job $j\in J$ with $d_{j}=\infty$
we can equivalently define $d_{j}=T$.

\subsection{Feasibility\label{sec:correct} }

For each DP-cell $(s,t,b)$ we computed a set of deadlines for the jobs in $J(s,t)$. We prove that for these deadlines
there is a schedule in which each job in $J(s,t)$ completes by its deadline. This holds in particular for the DP-cell $(0, T, 0)$
whose deadlines we used for the schedule that we output.

\begin{lemma}\label{lem:correctnessDP}
Let $\left\{ d_{j}\right\} _{j\in J}$
be the deadlines computed by the algorithm for some DP-cell $(s,t,b)$.
There exists a schedule for $J(s,t)$ starting at time $b$ in which each job $j$ completes by its deadline~$d_{j}$.
\end{lemma} 

\begin{lemma}\label{lem:makespan}
When we run EDF with the job deadlines stored in the DP-cell $(0, T, 0)$, then each job completes by time $T$ the latest.
\end{lemma}

\subsection{Approximation ratio}
We will show that our computed solution is a 6-approximation.
In the following, let $\OPT$ denote an optimal schedule and for each
job $j\in J$ let $d_{j}^{\OPT}$ be the completion time of $j$
in $\OPT$.
For each interval $[s,t)\in \I$
% $s, t$
there is one ``correct'' value of $b$ that
corresponds to $\OPT$.
We define that a cell $(s,t,b)$ is \emph{good} if $b$ is at most this value.
Intuitively, we will argue inductively that for
such good DP-cells our computed deadlines yield costs that are not much higher than the
corresponding costs in $\OPT$ for the jobs in $J(s,t)$.

\begin{definition}\label{def:goodDB} Consider a DP-cell $(s,t,b)$.
The cell $(s,t,b)$ is \emph{good} if $b$ is at most the total amount
of time during $[0, s]$ where the schedule $\OPT$ does \emph{not}
process jobs from $J(s,t)$.
\end{definition}

First, we prove a lemma for good DP-cells, which we will use to bound
the cost of the solution returned by the Lawler-Moore subroutine.
In this lemma, we define the sets $\leftjobs$ and $\rightjobs$
as in our algorithm above.

\begin{lemma}\label{lem:recLawler} Let $(s,t,b)$ be a good
DP-cell. Let $\leftjobs$ be the set of jobs $j\in J(s,t)$ for which \mbox{$r_j \le s - (t - s)$} and $\rightjobs = J(s,t)\setminus\leftjobs$. Let $\rightb$ be the total amount of time in $[0, s]$ during which $\OPT$ does not process jobs in $\rightjobs$.
Then there exist a schedule that processes all jobs $j\in\leftjobs$ with $d^{\OPT}_j \le s$ during~$[b, \rightb)$.
\end{lemma}
In the next lemma, we bound the costs due to the deadlines that we computed in the good DP-cells.
\begin{lemma}\label{lem:approx}
Consider a good DP-cell $(s, t, b)$. When we run EDF for the jobs $J(s,t)$
with the deadlines stored in $(s, t, b)$ starting at time $b$ , then the cost of the resulting schedule with respect to $\sum_j\cost_{j}(d_j)$ is bounded by
\[
2\hspace{-1em}\sum_{\substack{j\in J(s, t)\\
s<d_{j}^{\OPT}\le t
}
}\hspace{-1em}w_{j}(d_{j}^{\OPT}-r_{j})+4\sum_{j\in J(s, t)}w_{j}\cdot\len([r_{j},d_{j}^{\OPT}]\cap[s,t]).
\]
\end{lemma}

\begin{lemma}\label{lem:approxFactor} We output a schedule whose cost is
bounded by $6\cdot\OPT$.\end{lemma}
\begin{proof}
Consider the DP-cell $(0, T, 0)$, where $s = 0$ and $t = T$.
This is clearly a good DP-cell.
Thus, Lemma~\ref{lem:approx} yields an upper bound for
the cost of its computed schedule of 
\begin{align*}
 & 2\hspace{-1em}\sum_{\substack{j\in J(0, T) \\
0<d_{j}^{\OPT}\le T
}
}\hspace{-1em}w_{j}(d_{j}^{\OPT}-r_{j})+4\sum_{j\in J(0, T)}w_{j}\cdot\len([r_{j},d_{j}^{\OPT}]\cap[0,T])\\
 & \leq2\sum_{j\in J}w_{j}(d_{j}^{\OPT}-r_{j})+4\sum_{j\in J}w_{j}(d_{j}^{\OPT}-r_{j})\\
 & =6\sum_{j\in J}w_{j}(d_{j}^{\OPT}-r_{j})=6\cdot\OPT . \qedhere
\end{align*}
\end{proof}

\subsection{Running time}
Now we prove that our algorithm has pseudopolynomial running
time. Intuitively, this follows since the number of DP-cells is bounded
by $O(T^{3})$ and we try at most $O(T)$ values for $\rightb$ when we compute the solution for a DP-cell. 

\begin{lemma} The running time of the algorithm is bounded by $O(nT^{4})$.
\end{lemma}
\begin{proof}
Since the intervals $\mathcal I$ form a complete binary tree with $T$ leafs,
we have $|\mathcal I| \le O(T)$. Hence, there are $O(T)$ many
combinations $s, t$ for DP-cells. Since
also $b\in\mathbb Z$ with $0\le b\leq T$ there are at most $O(T^{2})$ many DP-cells.
When solving the DP, the algorithm tries $O(T)$ values for $\rightb$.
For each of those combinations,
we call the algorithm due to Lemma~\ref{lem:Lawler}, which has a
running time of $O(nT)$. The remaining operations are dominated by
this term. Thus, the total running time is at most
$O(nT^{4})$. 
\end{proof}
\begin{theorem}
The algorithm in this section is a pseudopolynomial 
time $6$-approximation algorithm
for minimizing sum of weighted flow times on a single machine with preemptions.
\end{theorem}

In Appendix~\ref{sec:polytime}, we will show how to transform the algorithm above into a \emph{polynomial} time
$(6+\epsilon)$-approximation algorithm. We use a similar argumentation as in the classical FPTAS for the knapsack problem, e.g., \cite{vazirani2001approximation}.
We round the costs to a discrete set of values of polynomial size. Then, we define a DP-cell for each combination of an interval $[s,t)\in \I$ (as before) and a target cost $B$ (instead of $b$). This yields a polynomial number of DP-cells. Then, for each resulting DP-cell $(s,t,B)$, we compute the best value of $b$ for which we find a set of deadlines that result in a cost of at most $B$.

Moreover, we show that we can generalize our algorithm to a polynomial time $O(1)$-approximation algorithm for minimizing
$(\sum_{j}w_{j}(F_{j})^{p})^{1/p}$ for any constant $p$ with $0<p<\infty$.
\begin{theorem}For any constant $p$ with $0<p<\infty$ there is a polynomial time $O(1)$-approximation algorithm
for minimizing $(\sum_{j}w_{j}(F_{j})^{p})^{1/p}$ on a single machine with preemptions.
\end{theorem}

\section{QPTAS for multiple machines}\label{sec:QPTAS}

In this section, we present a QPTAS for the problem of computing a
possibly preemptive schedule for a set of jobs $J$ on a set of $m$
unrelated machines $M$, where for each job $j\in J$ we are given
a release time $r_{j}\in\N_{0}$ and a processing time $p_{ij}\in\N\cup\{\infty\}$
for each machine $i\in[m]$, while minimizing the weighted $\ell_{p}$-norm
of the jobs' flow times, i.e., $(\sum_{j}w_{j}(C_{j}-r_{j})^{p})^{1/p}$
for some given $p$ with $0<p<\infty$. We obtain this result in the
settings with and without the possibility to migrate a job, i.e.,
to preempt it and resume it (later) on \emph{different }machine. For
any fixed $\varepsilon>0$ and $m$, our algorithm has quasi-polynomial
running time, assuming that all finite values $p_{ij}$ are quasi-polynomially
bounded in $n$. In the remainder of the section we assume that the
jobs are indexed by $J=\{1,2,\dotsc,n\}$ in non-decreasing order
of release times. Similarly as before, let $T:=\max_{j}r_{j}+np_{\max}$
which is hence an upper bound on the makespan of any schedule without
unnecessary idle times. We assume w.l.o.g.~that $1/\epsilon\in\N$.

We first discretize the time axis; specifically, in our computed schedule
all jobs will be started, preempted, and resumed only at (discrete)
time steps that are integral multiples of $\delta:=\varepsilon/n^{m}$.
Note that still jobs may still finish at non-discrete times if their
remaining processing volume needs less than $\delta$ units of time
on some machine to finish. We remark that $1/\delta\in\N$ and that
there are only a pseudopolynomial number of integral multiples of
$\delta$ in $[0,T]$. We prove in the next lemma that this discretization
can increase the optimal objective function value by at most a factor
of $1+\epsilon$.

\begin{lemma}\label{lem:chooseDelta} There exists a $(1+\varepsilon)$-approximative
schedule in which all jobs are started, preempted, and resumed only
at integral multiples of $\delta$. \end{lemma}
Next, for each job $j$ we construct a set $D(j)$ of $O(\log(T)/\varepsilon)$
potential deadlines such that between any two consecutive deadlines
the flow time of $j$ increases by at most a factor of $1+\epsilon$
(see Figure~\ref{fig:intervals}). This yields an interval for each pair of two consecutive
deadlines in $D(j)$. We include the release time $r_{j}$ in the
set $D(j)$ so that the first interval starts with $r_{j}$. Later,
we will decide for how much time we will execute $j$ on each machine
during each of these intervals; there are only $(n/\delta)^{O(m\log(T)/\varepsilon)}$
possibilities for this. We will ensure that the sets $D(j)$ for the
jobs $j\in J$ use a hierarchical structure, i.e., for each job $j$
the intervals for $j+1$ are a refinement of the intervals for $j$.
We note that a similar structure was used by Batra, Garg, and Kumar~\cite{Batra0K18},
however, the resulting loss was higher than only a factor of $1+\epsilon$.

\begin{lemma}\label{lem:deadlineProperties} In time polynomial
in $n$ and $\log(T)/\varepsilon$ we can construct a set of deadlines
$D(j)\subseteq\delta\mathbb{N}_{0}:=\{\delta k:k\in\N_{0}\}$ for
each job $j\in J$ with the following properties:
\begin{itemize}
\item for each job $j$ it holds that $|D(j)|\le O(\log(T)/\varepsilon)$,
\item for each job $j$ it holds that $r_{j}\in D(j)$, $r_{j}+1\in D(j)$,
and $T\in D(j)$,
\item for each job $j$ and two consecutive deadlines $d<d'$ with $d\geq r_{j}+1$
it holds that $\frac{d'-r_{j}}{d-r_{j}}\le1+\varepsilon$,
\item for each pair of consecutive jobs $j,j+1$, if $d\in D(j)$ and $d\ge r_{j+1}$,
then also $d\in D(j+1)$.
\end{itemize}
\end{lemma}
Consider a job $j\in J$ and assume that $D(j)=\{d_{0},d_{1},...,d_{k}\}$.
We define a set of intervals $\I(j)$ for $j$ by $\I(j):=\left\{ [d_{\ell},d_{\ell+1}):\ell\in\{0,...,k-1\}\right\} $.
Since $|D(j)|\le O(\log(T)/\varepsilon)$, we also have that $|\I(j)|\le O(\log(T)/\varepsilon)$.

The high level idea of our algorithm is that we consider the jobs
one by one, ordered non-decreasingly by their release times. For each
job $j$, we define for each interval in $I\in\I(j)$ for how long
$j$ is processed during $I$ on each machine. Assuming that $[d,d')\in\I(j)$
is the last interval during which we work on $j$, for the cost of
$j$ we take the value $w_{j}(d'-r_{j})^{p}$. As the processing time of each job is at least $1$ the third property
of Lemma~\ref{lem:deadlineProperties} yields that this is by at most a factor
of $(1+\epsilon)^{p}$ larger than the actual cost of $j$.

As indicated before, we assume that jobs can be started, preempted,
and resumed only at discrete times that are multiples of $\delta$
(although they may finish at non-discrete times). This means that
during any $\delta$-interval $[k\delta,(k+1)\delta]$ where $k\in\mathbb{N}$,
each machine can (partially) process only one job. Therefore, when
we consider a job $j$, there are only $(T/\delta)^{O(2^{m}\log(T)/\varepsilon)}$
possibilities for the number of free $\delta$-intervals on the $m$
machines during each interval in $\I(j)$ and for how this free time
is distributed among the $m$ machines. More precisely, this bounds
the number of possibilities for how many $\delta$-intervals each
exact set of machines is still free during each interval $I\in\I(j)$.
 Also, there are only $(T/\delta)^{O(2^{m}\log(T)/\varepsilon)}$
possibilities for defining for how many $\delta$-intervals we want
to work on $j$ during each interval $I\in\I(j)$, on which machine(s),
and which other machines are still free while we work on $j$. Since
the sets of deadlines $D(j)$ of the jobs form a hierarchical structure,
we can embed to above procedure into a dynamic program (DP) such that
for each of the $(T/\delta)^{O(2^{m}\log(T)/\varepsilon)}$ possible
decisions for scheduling $j$, there is a DP-cell for the resulting
subproblem for the next job $j+1$.

Formally, we define a DP-cell~$(j,L)$ for each combination of a
(next) job $j\in\{1,2,\dotsc,n\}$ and a load vector $L$, having
one entry $L(S,I)$ for each set of machines $S\subseteq M$ and each
$I\in\mathcal{I}(j)$. The entries of $L$ take values in $\{0,1,\dotsc,T/\delta\}$
and satisfy that $\sum_{S\subseteq M}L(S,I)\le|I|/\delta$ for every
$I\in\mathcal{I}(j)$ (we indicate by $|I|$ the length of $I$).
A DP-cell $(j,L)$ corresponds to the subproblem in which we want
to find a feasible allocation of $\delta$-intervals for all jobs
$\{j,j+1,\dotsc,n\}$, i.e., such that all these jobs are completely
processed. We impose the restriction that for each interval $I\in\mathcal{I}(j)$
and each machine set $S\subseteq M$, in $L(S,I)$ many $\delta$-intervals
during $I$ exactly the machines in $S$ are allocated to some jobs
from $\{j,\dotsc,n\}$. Intuitively, during these $\delta$-intervals,
each machine in $M\setminus S$ schedules a job in $\{1,...,j-1\}$
or is idle. The objective of the subproblem for the DP-cell $(j,L)$
is to minimize the resulting cost of the jobs in $\{j,...,n\}$.

We compute the solution for a DP-cell $(j,L)$ as follows: for each
$I\in\mathcal{I}(j)$, each set $S\subseteq M$, and each $i\in S$
we guess for how many $\delta$-intervals during $I$ the job $j$
is processed on $i$ and each machine in $S\setminus\{i\}$ processes
a job in $\{j+1,j+2,\dotsc,n\}$. Formally, we enumerate all guesses
and ultimately take the solution with lowest cost among all guesses.
For each combination of an interval $I\in\mathcal{I}(j)$, a set $S\subseteq M$,
and a machine $i\in S$ let $y(I,S,i)$ denote the guessed value.
We require that $L(I,S)\ge\sum_{i\in S}y(I,S,i)$ for each $I\in\mathcal{I}(j)$
and each set $S\subseteq M$ which means that the loads given by $L$
suffice. Note that we do not necessarily have equality, since there
might be $\delta$-intervals where no machine is allocated to $j$.
If job migrations are not allowed, we enumerate only guesses for which
$j$ is processed on only one machine. Further, we check that $j$
can be completed in its $\delta$-intervals, formally,
\begin{equation}
\sum_{i\in M}\sum_{I\in \I(j)}\sum_{S\subseteq M,i\in S}\delta\cdot y(I,S,i)/p_{ij}\ge1.\label{eq:dp-complete}
\end{equation}
From the guesses we can easily derive a residual load vector $L'$
that corresponds to the $\delta$-intervals allocated only to $\{j+1,j+2,\dotsc,n\}$,
more precisely,
\begin{equation}
L'(I,S):=L(I,S)-\sum_{i\in S}y(I,S,i)+\sum_{i\notin S}y(I,S\cup\{i\},i)\,\,\,\,\,\,\,\,\forall I\in\mathcal{I}(j)\,\,\forall S\subseteq M.\label{eq:dp-residual}
\end{equation}
Note that possibly the load vector $L'$ does not correspond to a
DP-cell for the next job $j+1$, since the intervals $\mathcal{I}(j)$
may differ from $\mathcal{I}(j+1)$. However, the latter intervals
form a subdivision of the former. Hence, it is straight-forward to
check for every DP-cell $(j+1,L'')$, whether $L''$ is ``compatible''
with $L'$. Formally, for each interval $I\in\mathcal{I}(j)$, let
$I_{1},\dotsc,I_{k}\in\mathcal{I}(j+1)$ be exactly the intervals
in $\mathcal{I}(j+1)$ contained in $I$, and we require for each
$S\subseteq M$ and each $i\in S$ that
\begin{equation}
L''(I_{1},S,i)+\cdots+L''(I_{k},S,i)=L'(I,S,i).\label{eq:dp-subdivision}
\end{equation}
We identify the DP-cell $(j+1,L'')$ with the solution of lowest cost
among all cells for the job $j+1$ for which the load vector $L''$
is compatible with $L'$. We extend it to a solution to the cell $(j,L)$
by assigning $\delta$-intervals to job $j$ according to the guessed
values $y(I,S,i)$, i.e., for each interval $I\in\I(j)$, each set
of machines $S\subseteq M$, and each $i\in S$, we take $y(I,S,i)$
many $\delta$-intervals during $I$ in which the machines $S\setminus\{i\}$
are busy with jobs in $\{j+1,...,n\}$ and schedule $j$ on $i$ during
these $\delta$-intervals.

In case that $j=n$, the job $j+1$ is not defined. In this case,
we allow only guesses for which the residual load vector $L'$ is
the all-zero vector. Then, instead of taking a precomputed schedule
from a DP-cell $(j+1,L'')$ like above, we take the schedule in which
no machine is working on any job.

At the end, we output a solution with lowest cost among all solutions
stored in the DP-cells of the form $(1,L)$ for some load vector $L$.

In the following lemma, we prove by induction over the jobs that our DP
computes the optimal solution to each DP-cell $(j,L)$.

\begin{lemma}\label{lem:dp} The dynamic program above computes the
optimal cost schedule where jobs can only be started, preempted, or
resumed at integral multiples of $\delta$ and the cost of a job $j$
is $w_{j}\cdot(d-r_{j})^{p}$ for the smallest $d\in D(j)$ with $d\ge C_{j}$.
\end{lemma}
Our running time is bounded by $O(n\log n)$ for sorting the jobs
initially by their release times, and by the number of DP-cells and
the number of guesses for them. The latter can be bounded by $n\cdot(T/\delta)^{O(2^{m}m\log(T)/\varepsilon)}$.

\begin{lemma}\label{lem:running-time}
The running time of our algorithm is bounded by $O(n\log n)+n\cdot(T/\delta)^{O(2^{m}m\log(T)/\varepsilon)}$.
\end{lemma}

\begin{theorem}\label{thm:QPTAS} For every $\varepsilon>0$ and every $p>0$ there
is a $(1+\varepsilon)$-approximation algorithm with a running time
of $O(n\log n)+(n\cdot p_{\max}/\varepsilon)^{O(2^{m}m\log(n\cdot p_{\max})/\varepsilon)}$
for minimizing $(\sum_{j}w_{j}(C_{j}-r_{j})^{p})^{1/p}$ on $m$ unrelated
machines in the settings with and without migration. \end{theorem}

\bibliographystyle{plain}
\bibliography{references}

\begin{thebibliography}{10}

\bibitem{armbruster2023ptas}
Alexander Armbruster, Lars Rohwedder, and Andreas Wiese.
\newblock A ptas for minimizing weighted flow time on a single machine.
\newblock In {\em Proceedings of {STOC}}, pages 1335--1344, 2023.

\bibitem{azar2018improved}
Yossi Azar and Noam Touitou.
\newblock Improved online algorithm for weighted flow time.
\newblock In {\em Proceedings of {FOCS}}, pages 427--437, 2018.

\bibitem{bansal2005minimizing}
Nikhil Bansal.
\newblock Minimizing flow time on a constant number of machines with
  preemption.
\newblock {\em Operations research letters}, 33(3):267--273, 2005.

\bibitem{bansal2009weighted}
Nikhil Bansal and Ho-Leung Chan.
\newblock Weighted flow time does not admit o(1)-competitive algorithms.
\newblock In {\em Proceedings of {SODA}}, pages 1238--1244, 2009.

\bibitem{bansal2007minimizing}
Nikhil Bansal and Kedar Dhamdhere.
\newblock Minimizing weighted flow time.
\newblock {\em {ACM} Transactions on Algorithms}, 3(4):39, 2007.

\bibitem{DBLP:journals/siamcomp/BansalP14}
Nikhil Bansal and Kirk Pruhs.
\newblock The geometry of scheduling.
\newblock {\em {SIAM} Journal on Computing}, 43(5):1684--1698, 2014.

\bibitem{Batra0K18}
Jatin Batra, Naveen Garg, and Amit Kumar.
\newblock Constant factor approximation algorithm for weighted flow time on a
  single machine in pseudo-polynomial time.
\newblock In {\em Proceedings of {FOCS}}, pages 778--789, 2018.

\bibitem{DBLP:journals/scheduling/BenderMR04}
Michael~A. Bender, S.~Muthukrishnan, and Rajmohan Rajaraman.
\newblock Approximation algorithms for average stretch scheduling.
\newblock {\em Journal of Scheduling}, 7(3):195--222, 2004.

\bibitem{chekuri2002approximation}
Chandra Chekuri and Sanjeev Khanna.
\newblock Approximation schemes for preemptive weighted flow time.
\newblock In {\em Proceedings of {STOC}}, pages 297--305, 2002.

\bibitem{chekuri2001algorithms}
Chandra Chekuri, Sanjeev Khanna, and An~Zhu.
\newblock Algorithms for minimizing weighted flow time.
\newblock In {\em Proceedings of {STOC}}, pages 84--93, 2001.

\bibitem{feige2019polynomial}
Uriel Feige, Janardhan Kulkarni, and Shi Li.
\newblock A polynomial time constant approximation for minimizing total
  weighted flow-time.
\newblock In {\em Proceedings of {SODA}}, pages 1585--1595, 2019.

\bibitem{lawler1969functional}
Eugene~L Lawler and J~Michael Moore.
\newblock A functional equation and its application to resource allocation and
  sequencing problems.
\newblock {\em Management science}, 16(1):77--84, 1969.

\bibitem{lenstraelements}
Jan~Karel Lenstra and David Shmoys.
\newblock Elements of scheduling.

\bibitem{RohwedderW21}
Lars Rohwedder and Andreas Wiese.
\newblock A {(2} + \emph{{\(\epsilon\)}})-approximation algorithm for
  preemptive weighted flow time on a single machine.
\newblock In {\em Proceedings of {STOC}}, pages 1042--1055, 2021.

\bibitem{schuurman-woeginger}
Petra Schuurman and Gerhard~J. Woeginger.
\newblock Polynomial time approximation algorithms for machine scheduling: ten
  open problems.
\newblock {\em Journal of Scheduling}, 2(5):203--213, 1999.

\bibitem{vazirani2001approximation}
Vijay~V Vazirani.
\newblock {\em Approximation algorithms}, volume~1.
\newblock Springer, 2001.

\bibitem{wiese:LIPIcs:2018:9432}
Andreas Wiese.
\newblock {Fixed-Parameter Approximation Schemes for Weighted Flowtime}.
\newblock In {\em Proceedings of {APPROX}}, pages 28:1--28:19, 2018.

\end{thebibliography}

\appendix

\section{Omitted proofs in Section~\ref{sec:algorithm}}
\subsection{Proof of Lemma~\ref{lem:Lawler}}
%\begin{lemma}[Lawler-Moore algorithm]
%\label{lem:Lawler} Given a set of $\widehat n$ jobs $\widehat{J}$ with release dates
%$\widehat{r}_j$, processing time $\widehat{p}_j$, cost $\widehat{c}_j$
%and a common deadline $\widehat{d}$. In time $O(\widehat n\sum_{j}\widehat{p}_j)$ we can
%compute a schedule with minimum cost for the above problem. 
%\end{lemma}
%
For the Lawler-Moore algorithm we refer to~\cite{lawler1969functional}.
The setting there is slightly different, hence we explain why it is equivalent.
In the algorithm, there is a common release date of $0$ and different
deadlines $\widehat{d}_j$, but we can rephrase our problem to bring it in
that form. For $j\in\widehat{J}$ let $\widehat{d}_j=\widehat{d}-\widehat{r}_j$. There is
the following one-to-one correspondence between the schedules with 
release date $0$ and deadlines  $\widehat d_{j}$ and the schedules with release dates
$\widehat{r}_j$ and common deadline $\widehat d$: when a job is executed in the time slot
$[s,s+1]$ with $s\in\mathbb{Z}$ in the first schedule, it is executed
in the interval $[\widehat d-(s+1),\widehat d-s]$ in the second schedule.
Thus, we can use the Lawler-Moore
algorithm with the deadlines dates $\widehat{d}_j$ and use the transformation above to get the desired schedule. 

\subsection{Proof of Lemma~\ref{lem:correctnessDP}}
%\begin{lemma}\label{lem:correctnessDP}
%Let $\left\{ d_{j}\right\} _{j\in J}$
%be the deadlines computed by the algorithm for some DP-cell $(s,t,b)$.
%There exists a schedule for $J(s,t)$, completing
%each job $j$ by its deadline~$d_{j}$.
%\end{lemma} 
%
We prove this via induction on $t-s$.
First consider the base cases where $t-s=1$. Then, each job
$j\in J(s,t)$ gets deadline $d_{j}=\infty$. Thus, there is nothing to show.

Now suppose that $t-s>1$. Fix $\rightb$ to be
the value of our guess that corresponds to the solution stored in $(s,t,b)$.
The algorithm uses the solutions in the DP-cells $(s,a,\rightb)$ and
$(a,t,\rightb)$. First, we prove the existence of these cells. 
Since $[s,t)\in\mathcal I$ and $a = (s + t)/2$ we have that $[s,a)\in\mathcal I$ and $[a,t)\in\mathcal I$.
By definition we have $\rightb \ge \max\{b,s-(t-s)\}$ and $\rightb\le s \le a$.
The earliest starting times of both $[s,a)$ and $[a,t)$ are $\max\{0,s - (t-s)\}$, which is therefore lower or equal $\rightb$. Thus $\rightb$ is in the eligible ranges and we conclude that both DP-cells exist.

Notice that the way we apply the Lawler-Moore algorithm, we
ensure that there is schedule for the jobs in
$\leftjobs$ with $d_{j} = s$ using only the time interval $[b,\rightb]$.
Thus it remains to show that there is a feasible schedule for the
jobs in $\rightjobs$ with deadlines $d_{j}$ using
only the interval $[\rightb,\infty)$.
We will use Lemma~\ref{lem:EDF} for this.
Since the lemma does not consider the restriction to an interval, i.e., to $[\rightb,\infty)$, we define modified release times
 $r_{j}'=\max\{r_{j}, \rightb\}$. By the lemma
we now only need to check that for each interval $[s',t']$
it holds that $\sum_{j\in \rightjobs: s'\leq r_{j}'\leq d_{j}\leq t'}p_{j}\leq t'-s'$. Let $d^{(1)}_j$ be the deadlines from the solution of
DP-cell $(s, a, \rightb)$ and $d^{(2)}_j$ those from $(a, t, \rightb)$.

First consider the case that $t'<a$. Since by induction hypothesis
there is a schedule with
release dates $r_{j}'$ and deadlines $d^{(1)}_j$, Lemma~\ref{lem:EDF}
implies $\sum_{j\in \rightjobs:s'\leq r_{j}'\leq d_{j}^{(1)}\leq t'}p_{j}\leq t'-s'$.
Each job $j$ with $d_{j}\leq t' < a$ fulfills $d_{j}\geq d_{j}^{(1)}$.
We conclude that
\begin{equation*}
\sum_{j\in \rightjobs:s'\leq r_{j}\leq d_{j}\leq t'}p_{j}\leq\sum_{j\in \rightjobs:s'\leq r_{j}\leq d_{j}^{(1)}\leq t'}p_{j}\leq t'-s'.
\end{equation*}
Similarly, when $t'\geq a$ each job with $d_{j}\leq t'$ fulfills
also $d_{j}^{(2)}\leq t'$. Together with Lemma~\ref{lem:EDF} applied
to the schedule with deadlines $d^{(2)}$ this proves $\sum_{j\in \rightjobs:s'\leq r_{j}\leq d_{j}\leq t'}p_{j}\leq t'-s'$

\subsection{Proof of Lemma~\ref{lem:makespan}}
%\begin{lemma}\label{lem:makespan}
%The schedule in DP-cell $(0, T, 0)$ completes each job before $T$.
%\end{lemma}
%
Recall that the schedule is produced by EDF. In particular, the machine
is never idle when there is a job that has been released, but not completed.
Thus, after time $\max_{j\in J} r_j$ the machine is never idle until
all jobs are completed. On the other hand it cannot be busy for longer
than $\sum_{j\in J} p_j$. Since $T \ge \max_{j\in J} r_j + \sum_{j\in J} p_j$ the statement of the lemma follows.

\subsection{Proof of Lemma~\ref{lem:recLawler}}
%\begin{lemma}\label{lem:recLawler} Let $(s,t,b)$ be a good
%DP-cell. As in the algorithm let $\leftjobs$ be the set of jobs $j\in J(s,t)$ for which $r_j \le s - (t - s)$ and $\rightjobs = J(s,t)\setminus\leftjobs$. Let $\rightb$ be the total amount of time in $[0, s]$ where schedule $\OPT$ does not process jobs in $\rightjobs$.
%Then there exist a schedule that processes all jobs $j\in\leftjobs$ with $d^{\OPT}_j \le s$ inside $[b, \rightb]$.
%\end{lemma} 
%
We may assume without loss of generality that $b$ is the total time
during $[0, s]$ where $\OPT$ does not process jobs in $J(s,t)$, since
this is the maximum value of $b$ for good DP-cell.
We show the existence of such a schedule using the third part of Lemma~\ref{lem:EDF} with release times $r'_j = \max\{b, r_j\}$ and
deadlines $d'_j = \rightb$ for all $j\in \leftjobs$ with $d^{\OPT}_j \le s$.
Consider an interval $[s',t']$ with $b\leq s'\leq t'\leq \rightb$.
Let $J' = \{j\in \leftjobs : d^{\OPT}_j \le s, s' \le r'_j, d'_j \le t'\}$.
If $J' = \emptyset$, the inequality
$\sum_{j\in J'} p_{j}\leq t'-s'$ is
trivially fulfilled. As all deadlines are $\rightb$, we can
assume $t'=\rightb$. Furthermore, because $\leftjobs$ contains only jobs released before $s-(t-s)$ we may assume $s'\leq s-(t-s)$.

First consider the case that $s'=b$. The values of $b$ and
$\rightb$ are chosen in such a way, that $\rightb-b$
is the amount of time $\OPT$ works during $[0,s]$ on jobs in $\leftjobs$.
As $\OPT$ completes every job in $J'$ within that
interval, we obtain $\sum_{j\in J'}p_{j}\leq\rightb - b = t' - s'$.

Now assume that $s'>b$. Then $r'_j = r_j$ for every job $j\in J'$.
Notice that $\OPT$ processes jobs from $\rightjobs$ for a duration of $s  - \rightb$ during $[s - (t - s), s]$,
in particular, during $[s', s]$.
Hence,
\begin{equation*}
\sum_{j\in J'}p_{j} \le \hspace{-1em} \sum_{j\in \leftjobs : s'\leq r_{j}, d^{\OPT}_j \le s} \hspace{-1em} p_{j} \le s - s' - (s - \rightb) = t' - s'. \qedhere
\end{equation*}

\subsection{Proof of Lemma~\ref{lem:approx}}
%\begin{lemma}\label{lem:approx} For each good DP-cell $(s, t, b)$
%the cost of the schedule computed by the algorithm is bounded by 
%\[
%2\hspace{-1em}\sum_{\substack{j\in J(s, t)\\
%s<d_{j}^{\OPT}\le t
%}
%}\hspace{-1em}w_{j}(d_{j}^{\OPT}-r_{j})+4\sum_{j\in J(s, t)}w_{j}\cdot\len([r_{j},d_{j}^{\OPT}]\cap[s,t]).
%\]
%\end{lemma} 
%
We prove the lemma by induction on the length of the interval
$[s,t)$. Let $(s,t,b)$ be a DP-cell and $a=(s+t)/2$.
Recall that in this case the algorithm
minimizes over all values of $\rightb$. Thus, we can
bound the cost of the schedule computed by the algorithm by the cost
of the schedule arising from specific values for $\rightb$.
Hence, assume without loss of generality that $\rightb$
is the total amount of time during $[0, s]$ where $\OPT$
does not process jobs in $\rightjobs$.
It is clear that $\rightb \le s$.
Furthermore, $\rightb \ge b$, since the DP-cell is a good cell and $\rightb \ge s - (t - s)$, since
no jobs from $\rightb$ are released (in particular, cannot be processed) before $s - (t - s)$.
Thus, $\rightb$ is
among the values tried by the algorithm.
We will now show the bound in the statement of the lemma separately for jobs in $\leftjobs$ and $\rightjobs$.

Consider first $\leftjobs$ and recall
that we use Lawler-Moore (Lemma~\ref{lem:Lawler}) to compute a schedule for those jobs.
Due to Lemma~\ref{lem:recLawler} we have
\begin{equation*}
\sum_{j\in \leftjobs}\cost_{j}(d_{j})\leq\sum_{j\in \leftjobs,d_{j}^{\OPT}>s} \hspace{-1em}\cost_{j}(t).
\end{equation*}
We will now bound the value of $\cost_{j}(t)$ for each job in the sum.
Let $j\in \leftjobs$ with $d_{j}^{\OPT}>s$. If $d_{j}^{\OPT}\leq t$
then $r_{j}\leq s-(t-s) < d^{\OPT}_j - (t - s)$ and
\begin{equation*}
t-r_{j}=t-s+s-r_{j}\leq t-s+d_{j}^{\OPT}-r_{j}\leq2(d_{j}^{\OPT}-r_{j}).
\end{equation*}
This implies $\cost_{j}(t)\leq2w_{j}(d_{j}^{\OPT}-r_{j})$. On the
other hand, if $d_{j}^{\OPT}>t$ then $\len([r_{j},d_{j}^{\OPT}]\cap[s,t])=t-s$.
Since $r_{j}\geq s-3(t-s)$ it follows $t-r_{j}\leq4(t-s)$. This
implies $\cost_{j}(t)\leq 4\cdot\len([r_{j},d_{j}^{\OPT}]\cap[s,t])$. Thus, 
\begin{align*}
\sum_{j\in \leftjobs}\cost_{j}(d_{j}) & \leq\sum_{\substack{j\in\leftjobs\\
s<d_{j}^{\OPT}\le t
}
}\hspace{-1em}\cost_{j}(t)+\sum_{j\in\leftjobs,d_{j}^{\OPT}>t} \hspace{-1em}\cost_{j}(t)\\
 & \leq2\hspace{-1em}\sum_{\substack{j\in \leftjobs\\
s<d_{j}^{\OPT}\le t
}
}\hspace{-1em}w_{j}(d_{j}^{\OPT}-r_{j})+4 \hspace{-1em}\sum_{j\in \leftjobs,d_{j}^{\OPT}>t}\hspace{-1em}\len([r_{j},d_{j}^{\OPT}]\cap[s,t])\\
 & \leq2\hspace{-1em}\sum_{\substack{j\in \leftjobs \\
s<d_{j}^{\OPT}\le t
}
}\hspace{-1em}w_{j}(d_{j}^{\OPT}-r_{j})+4\sum_{j\in \leftjobs}\len([r_{j},d_{j}^{\OPT}]\cap[s,t])
\end{align*}
It remains to bound the cost of the jobs in $\rightjobs$.
Consider a job $j\in \rightjobs$, i.e., $s - (t -s) < r_{j} < t$. 
If $t - s = 1$ then $j$ is released at $s$ and therefore $d^{\OPT}_j \ge t$.
Moreover, by construction $d_j = \infty$.
It follows that
\begin{equation*}
\sum_{j\in \rightjobs}\cost_{j}(d_{j})= \sum_{j\in \rightjobs} w_j (t - r_j) \le  2\hspace{-1em}\sum_{\substack{j\in \rightjobs\\
s<d_{j}^{\OPT}\le t
}
}w_{j}(d_{j}^{\OPT}-r_{j})+4\sum_{j\in \rightjobs}w_{j}\cdot\len([r_{j},d_{j}^{\OPT}]\cap[s,t]) .
\end{equation*}
Hence, assume for the remainder of the proof that $t - s > 1$.
In this case the solution is derived from the two DP lookups of cells
$(s, a, \rightb)$ and $(a, t, \rightb)$.
The first cell is good because $J(s,a) \subseteq \rightjobs$ and therefore
the total time during $[0, s]$ where the $\OPT$ does not process jobs from $J(s, a)$ is at
least $\rightb$. Similarly, the second cell is good because $J(a,t) = \rightjobs$ and the
time during $[0, a]$ where $\OPT$ does not process jobs from $J(a, t)$ is at least $\rightb$.

Let $d^{(1)}_j$ be the deadlines from the former cell and $d^{(2)}_j$
from the latter. Likewise, let $\cost_j^{(1)}$ and $\cost_j^{(2)}$
be the respective cost functions.
Consider now a job $j\in\rightjobs$. If
$d_{j}=d_{j}^{(2)}>a$, then $\cost_{j}^{(2)}(d_{j}^{(2)})=\cost_{j}(d_{j})$.
And if $d_{j}^{(2)}\leq a$, then $\cost_{j}^{(1)}(d_{j}^{(1)})=\cost_{j}(d_{j})$.
Thus, in both cases it holds $\cost_{j}(d_{j})\leq\cost_{j}^{(1)}(d_{j}^{(1)})+\cost_{j}^{(2)}(d_{j}^{(2)})$.
Together with the induction hypothesis we obtain: 
\begin{align*}
\sum_{j\in \rightjobs}\cost_{j}(d_{j})&= \sum_{j\in \rightjobs}\cost_{j}^{(1)}(d_{j}^{(1)})+\sum_{j\in \rightjobs}\cost_{j}^{(2)}(d_{j}^{(2)})\\
&=  2\hspace{-1em}\sum_{\substack{j\in \rightjobs\\ s<d_{j}^{\OPT}\le a}} \hspace{-1em} w_{j}(d_{j}^{\OPT}-r_{j})+4\sum_{j\in \rightjobs}w_{j}\cdot\len([r_{j},d_{j}^{\OPT}]\cap[s,a])\\
 &\quad +2\hspace{-1em}\sum_{\substack{j\in \rightjobs\\
a<d_{j}^{\OPT}\le t}} \hspace{-1em} w_{j}(d_{j}^{\OPT}-r_{j})+4\sum_{j\in \rightjobs}w_{j}\cdot\len([r_{j},d_{j}^{\OPT}]\cap[a,t])\\
&=  2\hspace{-1em}\sum_{\substack{j\in \rightjobs\\
s<d_{j}^{\OPT}\le t}}\hspace{-1em} w_{j}(d_{j}^{\OPT}-r_{j})+4\sum_{j\in \rightjobs}w_{j}\cdot\len([r_{j},d_{j}^{\OPT}]\cap[s,t])
\end{align*}
Adding together the inequalities for the costs of $\leftjobs$ and $\rightjobs$ finishes the proof.

\section{Polynomial time approximation for $p$-norms of weighted flow time\label{sec:polytime}}
The goal of this section is to  generalize the previous pseudopolynomial algorithm to objective functions $(\sum_j w_j(d_j-r_j)^p)^{1/p}$ with $0 < p < \infty$ and to turn it into a polynomial time algorithm for any fixed $p$. Throughout the algorithm's description and analysis we consider $\sum_j w_j(d_j-r_j)^p$ as the objective function, but then take the $p$-th root of the approximation factor in the end.
One obstacle for the running time is that so far there are $\Omega(T)$ DP-cells. Recall
that we have one DP-cell $(s,t,b)\in\mathbb{N}_{0}^{3}$
with $0\leq b\leq s<t\leq T$ for each $[s,t)\in\I$ and each
$b\ge b_{0}$ (where $b_{0}$ is defined based on $s$ and $t$).
First, we argue that we can reduce the number of intervals in $\I$
to $(n\log T)^{O(1)}$, using the fact that the intervals in $\I$ form a binary
tree and some simplifications. However, it is not clear how to round or restrict the starting
times $b$ to a polynomial number of options. To this end, we
use a method that is used in FPTASs for for the knapsack problem.
We round the costs of the jobs (depending on their flow times) such
that there is a polynomial number of options for the total cost. Then,
for each combination of an interval $[s,t)$ and a value for the total
cost, we compute a maximal value for the starting time $b$.

In the following we explain the changes in the algorithm and its analysis.
Let $\LB:=\sum_{j}w_{j}p_{j}^p$ be a lower bound on the optimal cost
(since each job has a flowtime of at least $p_{j}$). Let $\varepsilon>0$.
For any $x\geq0$ let $\lfloor x\rfloor_{\varepsilon/n\cdot\LB}$
denote the largest integer multiple of $\varepsilon/n\cdot\LB$ that
is at most $x$. 

\subsection{Polynomial time algorithm}

We introduce a DP-cell $(s,t,B)$ for each combination of an interval
$[s,t)\in\mathcal{I}$ for which a job is released during $[s-7(t-s),t+(t-s))$
and a budget $B$ which is a integer multiple of $\frac{\varepsilon}{n}\LB$
and which satisfies $0\leq B\leq(2^p+\frac{4^p}{4^p-3^p})n^p \cdot \LB$. 

The earliest start time $b_{0}$ and the job set $J(s,t)$ are defined
as before. The corresponding subproblem is to determine a starting
time $b$ with $b_{0}\leq b\leq s$ and deadlines for the
jobs in $J(s,t)$, as in Section~\ref{sec:algorithm}. We define 
\[
\cost_{j}(d_{j})=\begin{cases}
	0 & \text{ if }d_{j}=s\\
	\lfloor w_{j}\cdot(\min\{d_{j},t\}-r_{j})^p\rfloor_{\varepsilon/n\cdot\LB} & \text{ if }d_{j}>s\ .
\end{cases}
\]
Furthermore, we require that $\sum_{j\in J(s,t)}\cost_{j}(d_{j})\leq B$. We want to
compute a solution with a maximal starting time $b$ with this
property. It might happen that the a DP-cell is infeasible or that
the algorithm does not find a solution, e.g., because the budget $B$
is too small. We call such a DP-cell \emph{infeasible}. The DP-cells
where the algorithm computes a starting time are called \emph{feasible}.

After we computed solutions to all DP-cells, we determine the smallest budget $B$ for which the DP-cell $(0,T,B)$
is feasible. We consider the deadlines for the jobs according to $(0,T,B)$
and compute a schedule by running EDF with these deadlines. As before,
each job $j$ with $d_{j}=\infty$ is completed by time $T$.

The algorithm follows the same bottom-up structure as before. We have an additional base case due to the restriction that there is a job released in $[s-7(t-s),t+(t-s))$. In the case that $J(s,t)=\emptyset$, there are no jobs. Thus we choose the maximal starting time possible, i.e. $b=s$.  For the remainder we can assume $J(s,t)\neq \emptyset$.
Let $\leftjobs$ as before be the set of jobs $j\in J(s,t)$ with $r_j\leq s-(t-s)$ and $\rightjobs=J(s,t)\setminus \leftjobs$. 
The cells $(s,t,B)$ for which $t-s=1$ are treated separately. 

First, we consider the case $t-s>1$. Assume that we are given such a cell $(s,t,B)$. We enumerate all possibilities for three budgets $B_{0},B_{1},B_{2}\geq0$,
each being an integral multiple of $\frac{\varepsilon}{n}\LB$ and
for which $B_{0}+B_{1}+B_{2}\leq B$ holds. Set $a=\frac{s+t}{2}$. We will use the deadlines
computed in the DP-cells $(s,a,B_{1})$ and $(a,t,B_{2})$. First,
we show that these cells indeed exist.
\begin{lemma}
	Consider a DP-cell $(s,t,B)$ for which $t-s>1$ and $J(s,t)\neq\emptyset$.
	Let $B_{1},B_{2}$ be integral multiples of $\frac{\varepsilon}{n}\LB$
	for which $B_{0}+B_{1}+B_{2}\leq B$ holds. Then $(s,a,B_{1})$ and
	$(a,t,B_{2})$ are DP-cells.
\end{lemma}

\begin{proof}
	Its is clear that $B_{1},B_{2}\leq(2^p+\frac{4^p}{4^p-3^p})n^p\LB$ and that $[s,a),[a,t)\in\mathcal{I}$.
	As $J(s,t)\neq\emptyset$ there must be a job release in $[b_{0},t)$.
	As $s-7(a-s)=s-3.5(t-s)<a-7(t-a)=s-3(t-s)\leq b_{0}$ and $a+(a-s)=t$,
	this job is released in $[s-7(a-s),s+(a-s))$ and $[a-7(t-a),t+(t-a))$.
	Thus the cells $(s,a,B_{1})$ and $(a,t,B_{2})$ exist for any budgets
	$B_{1},B_{2}$ which are multiples of $\frac{\varepsilon}{n}\LB$
	and for which $B_{0}+B_{1}+B_{2}\leq B$.
\end{proof}
If one of the DP-cells $(s,m,B_{1})$ or $(m,t,B_{2})$ is infeasible,
we simply continue with the next combination for $B_{0},B_{1},B_{2}$.
Otherwise, we obtain sets of deadlines $d^{(1)}$ and $d^{(2)}$ as
well as starting times $b^{(1)}$ and $b^{(2)}$ from the
computed solutions for $(s,a,B_{1})$ and $(a,t,B_{2})$, respectively.
Let  $\rightb:=\min\{b^{(1)},b^{(2)}\}$.
For each job $j\in \rightjobs $ we define $d_{j}$ as 
\[
d_{j}:=\begin{cases}
	d_{j}^{(2)} & \text{ if }d_{j}^{(2)}>a\\
	\min\{d_{j}^{(1)},a\} & \text{ otherwise.}
\end{cases}
\]
This yields deadlines corresponding to a schedule for $\rightjobs$ with
starting time $\rightb$ and cost at most $B_{1}+B_{2}$. 

Now consider a cell $(s,t,B)$ with $t-s=1$. No job $j \in \rightjobs$ can be completed before $t$ as $r_j+p_j\geq r_j+1>s-(t-s)+1=s$ implies $r_j+p_j\geq s+1=t$.
So there can only exists a solution if $B\geq\sum_{j\in \rightjobs}\cost_{j}(\infty)$. In this case let $B_{0}:=B-\sum_{j\in \rightjobs}\cost_{j}(\infty)$ be the remaining budget for jobs in $\leftjobs$. Furthermore let $\rightb:=s$ as there is no need to start with $\rightjobs$ before $s$. 

No matter whether $t-s=1$ or $t-s>1$, we need to define deadlines for the jobs in $\leftjobs$. Like before, we do this
such that $d_{j}\in\{s,\infty\}$ for each job $j\in \leftjobs$. Specifically, we use a modified version of the Lawler-Moore algorithm that runs in polynomial time to solve the following problem: we are given a set of $\widehat{n}$ Jobs $\widehat{J}$ where each $j \in \widehat{J}$ is characterized by a processing time $\widehat{p}_j\in \N$, a release date $\widehat{r}_j\in \N$ and a penalty cost $\widehat{c}_j$ where $\widehat{c}_j$ is a multiple of $\varepsilon/n\cdot LB$. In addition there is a common due date $\widehat{d}\in \N$ and a budget $\widehat{B}$, also a multiple of $\varepsilon/n\cdot LB$. The goal is  to compute a latest starting
time $\widehat{b}$  and a schedule for $\widehat{J}$ with deadlines $\widehat{d}_j$ where no job is scheduled before $\widehat{b}$ and the costs are bounded by $\widehat{B}$, i.e. $\sum_{j\in \widehat{J}: \widehat{d}_j>\widehat{d}}\widehat{c}_j\leq \widehat{B}$ or to decide that such a starting time does not exist.

\begin{lemma}\label{lem:modifiedLawler} There is an algorithm computing
	the optimal solution for the above problem in a running time that
	is polynomial in $\widehat{n}$ and $\widehat{B}/(\varepsilon/n\cdot\LB)$. 
\end{lemma} 
\begin{proof}
	This can be achieved by a variant of the Lawler-Moore algorithm. We refer
	to~\cite{lenstraelements}, Chapter~5.2 for a detailed description. Note that
	the reference uses integer weights, which can be achieved by scaling appropriately,
	and an inverted role of release dates and deadlines, as explained in the proof of Lemma~\ref{lem:Lawler}.
\end{proof}
We solve the problem above for the jobs $\widehat{J}=\leftjobs$ with release time $\widehat{r}_j=r_j$, processing times $\widehat{p}_j=p_j$ and costs $\widehat{c}_j=\cost_{j}(\infty)$. The due date is $\widehat{d}=\rightb$ and the budget is $\widehat{B}=B_0$.
In the case that $t-s=1$, the resulting schedule is the schedule for the DP-cell $(s,t,B)$ with $b=\widehat{b}$ and the DP-cell is declared infeasible if the modified Lawler-Moore algorithm does not find a schedule.

Now consider $t-s>1$. If the modified Lawler-Moore algorithm does not yield a schedule and a starting time $\widehat{b}$, we do not consider
the combination of budgets $B_{0},B_{1},B_{2}$ further. Otherwise we use the computed deadlines for $\leftjobs$ and define $b=\widehat{b}$ as the starting time for this problem. Finally
we choose the schedule with the maximal starting $b$ among all
combinations of $B_{0},B_{1},B_{2}$ for which we obtained a schedule.
If there is no such combination for $B_{0},B_{1},B_{2}$, then the
DP-cell $(s,t,B)$ is declared infeasible.

In the end, we determine the minimal budget $B$ for which the DP-cell $(0,T,B)$
is feasible. We consider the deadlines for the jobs according to $(0,T,B)$
and compute a schedule by running EDF with these deadlines. As before,
each job $j$ with $d_{j}=\infty$ is completed by time $T$.

\subsection{Polynomial time approximation factor}

In comparison to our pseudopolynomial time routine, we loose a factor
of $1+\varepsilon$ due to the rounding of the cost functions. The
proof of the next lemma follows the same structure as our argumentaiton
for the pseudopolynomial time algorithm and most of the inequalities
appear similarly there. 
\begin{lemma}\label{lem:polyApproxFactor}
	The algorithm in this section returns a $\big((2^p+\frac{4^p}{4^p-3^p})^{1/p}+\varepsilon\big)$-approximation. 
\end{lemma} 
\begin{proof}
	The following definition is inspired by Definition~\ref{def:goodDB}
	and Lemma~\ref{lem:approx}, adjusted to our new definition of DP-cells
	using a budget instead of a starting time. For an interval $[a,b]$ let $\vorne ([a,b]):=a$ and $\hinten ([a,b]):=b$ and $\vorne (\emptyset)=\hinten (\emptyset):=T$. Further, let $C:=\frac{4^p}{4^p-3^p}>0$. Let $\OPT$ be an optimal schedule with deadlines $d_{j}^{\OPT}$ and, slightly abusing notation, we also use to to refer to the cost $\sum_{j \in J}w_j(d_j^{\OPT}-r_j)^p$. We call a DP-cell $(s,t,B)$
	\emph{good} if 
	\begin{multline*}
	B\geq \Big\lfloor2^p\hspace{-1em}\sum_{\substack{j\in J(s,t)\\
			s<d_{j}^{\OPT}\le t
		}
	}\hspace{-1em}w_{j}(d_{j}^{\OPT}-r_{j})^p+C\sum_{j\in J(s,t)}w_{j}\cdot\Big(\bigl(\hinten([r_{j},d_{j}^{\OPT}]\cap[s,t])-r_j\bigr)^p \\[-2em]
- \big( \vorne([r_{j},d_{j}^{\OPT}]\cap[s,t])-r_j\big)^p\Big) \Big\rfloor_{\varepsilon/n\cdot\LB} \ .
	\end{multline*}
	Note that for $p=1$ this is exactly the bound from  Lemma~\ref{lem:approx}.
	We prove by induction on $t-s$ that each good DP-cell is feasible and
	its computed starting time is at least $b^{\OPT}$,
	the total amount of time during $[0,s]$ the schedule $\OPT$ does not process jobs from $J(s,t)$. To this end, let $(s,t,B)$ be a good DP-cell.
	
	First consider the base case: if
	$J(s,t)=\emptyset$, the algorithm chooses $b:=s$ and thus the cell is good for any budget. Now suppose that $J(s,t)\neq \emptyset$. 
	In the case $t-s>1$, we maximize the starting time over all budgets $B_{0},B_{1},B_{2}$,
	thus we can bound the starting time by choosing specific values
	for $B_{0},B_{1},B_{2}$. Hence, assume without loss of generality that 
	\begin{multline*}
		B_{0}  =\Big\lfloor2^p\hspace{-1em}\sum_{\substack{j\in \leftjobs\\
				s<d_{j}^{\OPT}\le t
			}
		}\hspace{-1em}w_{j}(d_{j}^{\OPT}-r_{j})^p+C\sum_{j\in \leftjobs}w_{j}\cdot \Big( \big(\hinten([r_{j},d_{j}^{\OPT}]\cap[s,t])-r_j\big)^p \\[-2em]
-\big(\vorne([r_{j},d_{j}^{\OPT}]\cap[s,t])-r_j\big)^p \Big) \Big\rfloor_{\varepsilon/n\cdot\LB} \ ,
	\end{multline*}
	\begin{multline*}
		B_{1}  =\Big\lfloor2^p\hspace{-1em}\sum_{\substack{j\in \rightjobs\\
				s<d_{j}^{\OPT}\le a
			}
		}\hspace{-1em}w_{j}(d_{j}^{\OPT}-r_{j})^p+C\sum_{j\in \rightjobs}w_{j}\cdot\Big( \big(\hinten([r_{j},d_{j}^{\OPT}]\cap[s,a])-r_j\big)^p \\[-2em]
-\big(\vorne([r_{j},d_{j}^{\OPT}]\cap[s,a])-r_j\big)^p \Big) \Big\rfloor_{\varepsilon/n\cdot\LB} \ ,
	\end{multline*}
	\begin{multline*}
		B_{2}  =\Big\lfloor2^p\hspace{-1em}\sum_{\substack{j\in \rightjobs\\
				a<d_{j}^{\OPT}\le t
			}
		}\hspace{-1em}w_{j}(d_{j}^{\OPT}-r_{j})^p+C\sum_{j\in \rightjobs}w_{j}\cdot \Big( \big(\hinten([r_{j},d_{j}^{\OPT}]\cap[a,t])-r_j\big)^p \\[-2em]
-\big(\vorne([r_{j},d_{j}^{\OPT}]\cap[a,t])-r_j\big)^p \Big) \Big\rfloor_{\varepsilon/n\cdot\LB} \ .
	\end{multline*}
	Note that $B_{0}+B_{1}+B_{2}\leq B$. Then the cells
	$(s,a,B_{1})$ and $(a,t,B_{2})$ exist and are good.
	Thus, by the induction hypothesis they are solved with the starting
	times $b^{(1)}$ and $b^{(2)}$, where $b^{(1)}$ is at least the
	total amount of time during $[0,s]$ the schedule $\OPT$ does not process jobs from $J(s,a)\subseteq\rightjobs$ (in particular, at least the time $\OPT$ does not process jobs from $\rightjobs$) and $b^{(2)}$ is at least the total amount
	of time during $[0,a]$ the schedule $\OPT$ does not process jobs from $J(a,t)=\rightjobs$.
	Let $\rightb=\min\{b^{(1)},b^{(2)}\}$.
	
	In the case that $t-s=1$ we define $B_0$ as above and $\rightb:=s$. Note that $B_0\leq B$.
	The remaining argument is the same no matter whether $t-s=1$ or $t-s>1$. In Lemma~\ref{lem:recLawler} was shown that there is a schedule using only the interval $[b^{\OPT},\rightb]$
	and completing each job $j\in \leftjobs$ with $d_{j}^{\OPT}\leq s$.
	Next we show that $B_{0}\geq\sum_{j\in \leftjobs:d_{j}^{\OPT}>s}\cost_{j}(\infty)$.
	
	Let $j \in \leftjobs$ with $d_j^{\OPT}>s$. First consider $d_j^{\OPT}\leq t$: As $r_j\leq s-(t-s)$ we obtain $t-r_j\leq (t-s)+(s-r_j)\leq 2(d_j^{\OPT}-r_j)$ and therefore $\cost_{j}(\infty)\leq w_j(t-r_j)^p\leq w_j\cdot (2(d_j^{OPT}-r_j))^p=2^pw_j(d_j^{\OPT}-r_j)^p$.
	
	Now let $d_j^{\OPT}> t$. Then $\hinten([r_{j},d_{j}^{\OPT}]\cap[s,t])=t$ and as $j \in \leftjobs$ also $\vorne([r_{j},d_{j}^{\OPT}]\cap[s,t])=s$. As $r_j\geq s-3(t-s)=4s-3t$ it follows $3t-3r_j\geq 4s-4r_j$ and thus $4^p(s-r_j)^p\leq 3^p(t-r_j)^p$. This yields \begin{align*}
		\cost_{j}(\infty)&\leq w_j(t-r_j)^p\\
		&\leq w_j\frac{1}{4^p-3^p}(4^p(t-r_j)^p-3^p(t-r_j)^p)\\
		&\leq w_j\frac{1}{4^p-3^p}(4^p(t-r_j)^p-4^p(s-r_j)^p)\\
		&=w_j \frac{4^p}{4^p-3^p}((t-r_j)^p-(s-r_j)^p)\\
		&=C\cdot w_{j}\cdot \Big( \big(\hinten([r_{j},d_{j}^{\OPT}]\cap[s,t])-r_j\big)^p-\big(\vorne([r_{j},d_{j}^{\OPT}]\cap[s,t])-r_j\big)^p \Big)
	\end{align*}
	Together with the case $d_j^{\OPT}\leq t$ this yields
	$$\sum_{j\in \leftjobs:d_{j}^{\OPT}>s}\cost_{j}(\infty)\leq B_0.$$
	This implies that the budget $B_{0}$ is sufficiently large to find that schedule. 
	Thus we get a schedule with a starting time of at least $b^{\OPT}$.
	This completes the proof that a good DP-cell has a starting time of
	at least~$b^{\OPT}$.
	
	Before proving the approximation ratio, we need an upper bound on
	$\OPT$ to show that there is a good cell for the interval $[0,T)$.
	The shortest remaining processing time rule can be applied to the
	weighted flow time problem (i.e. always working on the job that has
	currently the shortest remaining processing time). Consider a job
	$j$. Between the release time $r_{j}$ and the deadline $d_{j}$
	we only work on jobs with remaining processing time of at most $p_{j}$.
	As there are $n$ job, we can only work for a period of $n\cdot p_{j}$.
	Thus $d_{j}-r_{j}\leq n\cdot p_{j}$. We conclude that the resulting schedule
	has a cost of at most $\sum_{j\in J}w_j(d_j-r_j)^p\leq \sum_{j\in J}w_j\cdot n^p \cdot p_j^p \leq n^p\cdot \LB$.
	This implies that $\OPT\leq n^p\cdot\LB$.
	
	In the end the algorithm outputs the solution of the DP-cell $(0,T,B)$
	with minimal $B$ for which the cell is feasible. A DP-cell $(0,T,B)$
	is good if 
	\begin{align*}
		B & =\Big\lfloor2^p\hspace{-1em}\sum_{\substack{j\in J(0,T)\\
				0<d_{j}^{\OPT}\le T
			}
		}\hspace{-1em}w_{j}(d_{j}^{\OPT}-r_{j})^p+C\sum_{j\in J'}w_{j}\cdot\Big( \big(\hinten([r_{j},d_{j}^{\OPT}]\cap[0,T])-r_j\big)^p \\[-2em]
 & & \hspace{-2in} - \big(\vorne([r_{j},d_{j}^{\OPT}]\cap[0,T])-r_j\big)^p \Big) \Big\rfloor_{\varepsilon/n\cdot\LB}\\
		& =\Big\lfloor2^p\sum_{j\in J}w_{j}(d_{j}^{\OPT}-r_{j})^p+C\sum_{j\in J}w_{j}\cdot(d_{j}^{\OPT}-r_{j})^p\Big\rfloor_{\varepsilon/n\cdot\LB}\\
		& =\Big\lfloor(2^p+\frac{4^p}{4^p-3^p})\cdot\OPT\Big\rfloor_{\varepsilon/n\cdot\LB} \ .
	\end{align*}
	As $\lfloor (2^p+\frac{4^p}{4^p-3^p})\cdot\OPT\rfloor_{\varepsilon/n\cdot\LB}\leq (2^p+\frac{4^p}{4^p-3^p}) \cdot\OPT \leq (2^p+\frac{4^p}{4^p-3^p}) \cdot n^p \cdot \LB$
	this shows that there is a good DP-cell $(0,T,B)$ with $B\leq (2^p+\frac{4^p}{4^p-3^p}) \cdot\OPT$.
	Since all good DP-cells are solved, the algorithm outputs a schedule
	corresponding to a cell with budget at most $B\leq (2^p+\frac{4^p}{4^p-3^p}) \cdot\OPT$. The
	actual weighted flow time might be bigger than $B$ as we rounded
	down the cost in our DP. Thus the final bound on the weighted p-norm is
	\begin{align*}
		\big(\sum_{j\in J}w_{j}(d_{j}-r_{j})^p \big)^{1/p} & \leq \big(\sum_{j\in J}(\lfloor w_{j}(d_{j}-r_{j})^p\rfloor_{\varepsilon/n\cdot\LB}+\varepsilon/n\cdot\LB) \big)^{1/p}\\
		& \leq (B+\varepsilon \LB)^{1/p}\\
		& \leq \big((2^p+\frac{4^p}{4^p-3^p}+\varepsilon)\OPT\big)^{1/p}\\
		& \leq \big((2^p+\frac{4^p}{4^p-3^p})^{1/p}+\varepsilon\big)\cdot\big(\sum_{j\in J}w_{j}(d_{j}^{\OPT}-r_{j})^p \big)^{1/p}
	\end{align*}
	Thus, our the algorithm computes a $\big((2^p+\frac{4^p}{4^p-3^p})^{1/p}+\varepsilon\big)$-approximation for weighted $p$-norm.
\end{proof}

\subsection{Polynomial running time}

In this section we prove that with the changes made to the algorithm
we indeed have a polynomial running time. 

\begin{lemma}\label{lem:polyRunningTime} The number of DP-cells
	and the running time of our algorithm is bounded by $(n+\log T)^{O(1)}$.
\end{lemma} 
\begin{proof}
	First, we prove that there are only a polynomial number of DP-cells.
	Consider a DP-cell $(s,t,B)$. We bound the number of different budgets.
	As $0\leq B\leq(2^p+\frac{4^p}{4^p-3^p})\cdot n^p\cdot \LB $ and $B$ is a multiple of $\varepsilon/n\cdot\LB$
	there are at most $(2^p+\frac{4^p}{4^p-3^p})\cdot n^{p+1}/\varepsilon$ many options for $B$.
	
	Next, we bound the number of different intervals $[s,t)\in\mathcal{I}$
	for which a corresponding DP-cell can exist (for some budget $B$).
	When the cell $(s,t,B)$ exists, there must be a job $j$ released
	during $[s-7(t-s),t+(t-s))$. As the intervals in $\mathcal{I}$ form
	a binary tree, if $[s,t)\in\mathcal{I}$ then $t-s=2^{k}$ for some
	integer $k$ for which $0\leq k\leq\log_{2}(T)$ holds. Furthermore,
	$s$ and $t$ are integer multiples of $2^{k}$. So for fixed $k$
	and $j$ there are at most $9$ intervals $(s,t)\in\mathcal{I}$ with
	$t-s=2^{k}$ and $r_{j}\in[s-7(t-s),t+(t-s))$. As there are $n$
	jobs and $\log_{2}(T)+1$ possible values for $k$, there are at most
	$9n(\log_{2}(T)+1)$ intervals $[s,t)\in\mathcal{I}$ for which a
	corresponding DP-cell can exist. 
	
	Now we bound the running time of our algorithm. As we have already
	bounded the number of DP-cells by $(n+\log T)^{O(1)}$, it remains
	to bound the needed time to process one DP-cell. Recall that we enumerate
	budgets $B_{0},B_{1},B_{2}$ that are all integral multiples of $\varepsilon/n\cdot\LB$
	and for which $B_{0}+B_{1}+B_{2}\leq B\leq (2^p+\frac{4^p}{4^p-3^p})\cdot n^p \cdot \LB $ holds. Thus, there
	are at most $O(n^{p+1}/\varepsilon)$ many options for each value $B_{0}$,
	$B_{1}$ and $B_{2}$, and hence at most $O(n^{3p+3}/\varepsilon^{3})$
	many combinations for these values. For each of these combinations
	we called the modified version of the Lawler-Moore algorithm (see
	Lemma~\ref{lem:modifiedLawler}) that also runs in polynomial time.
	This implies that our algorithm runs in time $(n+\log T)^{O(1)}$.
\end{proof}
\begin{theorem} Our algorithm is a polynomial time $\big((2^p+\frac{4^p}{4^p-3^p})^{1/p}+\varepsilon\big)$-approximation
	algorithm for the weighted $p$-norm of flow time problem on a single machine when
	preemptions are allowed. \end{theorem} 
\begin{proof}
	As shown in Lemma~\ref{lem:polyApproxFactor}, the algorithm is a
	$\big((2^p+\frac{4^p}{4^p-3^p})^{1/p}+\varepsilon\big)$-approximation. And as shown in Lemma~\ref{lem:polyRunningTime}
	the algorithm runs in polynomial time. 
\end{proof}

\section{Omitted proofs in Section~\ref{sec:QPTAS}}

\subsection{Proof of Lemma~\ref{lem:chooseDelta}}
% \begin{lemma}\label{lem:chooseDelta} There exists a $(1+\varepsilon)$-approximative
% schedule in which all jobs are started, preempted, and resumed only
% at integral multiples of $\delta$. \end{lemma}
% \begin{proof}
First, we show that we can delay the completion of each job by $\varepsilon$
at a cost of at most $1+\varepsilon$. Let $j$ be a job. Then $C_{j}-r_{j}\ge p_{j}\ge1$.
Thus, $(C_{j}-r_{j}+\varepsilon)^{p}\le(1+\varepsilon)^{p}(C_{j}-r_{j})^{p}$.
Since we take the $p$-th root of the sum over all weighted flow times,
the exponent of $(1+\epsilon)$ disappears.

For a tupel $(j_{1},\dots,j_{m})\in(J\cup\{\bot\})^{m}$ of $m$ jobs
(or placeholder $\bot$) let $\ell(j_{1},\dots,j_{m})$ denote the
total amount of time in $\OPT$ where simultaneously for each $i$
the job $j_{i}$ is executed on machine $i$. The symbol $\bot$ is
used for idle time. For intuition, we now imagine a single machine
scheduling problem with a ``tuple-job'' of length $\ell(j_{k_{1}},\dots,j_{k_{m}})$
for every $m$-tuple of jobs, which is released at $\max\{r_{j_{k_{1}}},\dotsc,r_{j_{k_{m}}}\}$
and whose deadline is $\min\{C_{j_{k_{1}}},\dotsc,C_{j_{k_{m}}}\}$,
where $C_{j}$ are the completion times in $\OPT$. The optimal schedule
asserts that there exists a feasible schedule for this instance and,
thus, we may also assume that it is obtained from the Earliest-deadline-first
(EDF) algorithm w.r.t. the tuple-jobs. If we now increase the processing
time of each job to the next integral multiple of~$\delta$, then
EDF produces a schedule where all jobs are started, preempted, and
resumed only at integral multiples of $\delta$, which also translates
to a discrete schedule of the original instance. Note that we may
process some jobs longer than necessary. Also, note that a job might
not finish by its original deadline. However, the accumulated delay
of some tuple-job is at most the number of tuple-jobs that have a
smaller deadline, multiplied by $\delta$ (since we rounded up the
processing time of each tuple-job by at most $\delta$). Hence, the
total delay of each job $j$ is bounded by at most $n^{m}\cdot\delta\le\varepsilon$.
% \end{proof}

\subsection{Proof of Lemma~\ref{lem:deadlineProperties}}
% \begin{lemma}\label{lem:deadlineProperties} In time polynomial
% in $n$ and $\log(T)/\varepsilon$ we can construct a set of deadlines
% $D(j)\subseteq\delta\mathbb{N}_{0}:=\{\delta k:k\in\N_{0}\}$ for
% each job $j\in J$ with the following properties:
% \begin{itemize}
% \item for each job $j$ it holds that $|D(j)|\le O(\log(T)/\varepsilon)$,
% \item for each job $j$ it holds that $r_{j}\in D(j)$, $r_{j}+1\in D(j)$,
% and $T\in D(j)$,
% \item for each job $j$ and two consecutive deadlines $d<d'$ with $d<r_{j}$
% it holds that $\frac{d'-r_{j}}{d-r_{j}}\le1+\varepsilon$,
% \item for each pair of consecutive jobs $j,j+1$, if $d\in D(j)$ and $d\ge r_{j+1}$,
% then also $d\in D(j+1)$.
% \end{itemize}
% \end{lemma}
% \begin{proof}
In the construction, we first define non-discrete
deadlines, which we later round to integral multiples of $\delta$.
We iterate over all jobs $j=1,2,\dotsc,n$ in this order, in particular,
in non-decreasing order of release times. Assume we already constructed
the deadlines for all jobs $\{1,...,j-1\}$. If $j$ is the first
job, we initialize $D(j)=\{r_{j},r_{j}+1,T\}$, otherwise we set $D(j)=\{r_{j},r_{j}+1\}\cup(D(j-1)\cap[r_{j},\infty))$,
that is, we take the deadlines of the previous job restricted to the
relevant range for $j$ and add $r_{j}$ and $r_{j}+1$. We now search
for a violation of the third property, that is, we search for two
consecutive deadlines $d,d'\in D(j)$ with $ r_{j}+1\leq d<d'$ and $(d'-r_{j})/(d-r_{j})>1+\varepsilon$.
If such deadlines exist, we introduce a new one by taking a geometric
average. More precisely, we insert deadline $d''$ with $d''=r_{j}+\sqrt{(d-r_{j})\cdot(d'-r_{j})}$.
We repeat until exhaustion. After creating the non-discrete deadlines
for all jobs, we replace for each job each potential deadline $d$
by two discrete deadlines $\lfloor d \rfloor_\delta$ and $\lceil d \rceil_\delta$, i.e.,
the largest integral multiple of
$\delta$ that is at most $d$ and the smallest integral multiple
that is at least~$d$.

We will first consider the non-discrete deadlines and argue that the
procedure terminates and for any job $j$ two consecutive deadlines
$d'>d> r_{j}+1$ of $j$ satisfy
$(d'-r_{j})/(d-r_{j})\ge\sqrt{1+\varepsilon}$. Consider the first
job $j=1$. The potential deadlines are initialized as $\{r_{1},r_{1}+1,T\}$.
At this time, the property is trivially satisfied, since there is
no such pair of deadlines. Whenever the procedure adds a new deadlines
$d''$ between $d$ and $d'$, then $d'-r_{1}>(1+\varepsilon)(d-r_{1})$
and $d''-r_{1}=\sqrt{(d'-r_{1})\cdot(d-r_{1})}$. Thus,
\[
\frac{d'-r_{1}}{d''-r_{1}}=\sqrt{\frac{d'-r_{1}}{d-r_{1}}}\ge\sqrt{1+\varepsilon}\ .
\]
Similarly,
\[
\frac{d''-r_{1}}{d-r_{1}}=\sqrt{\frac{d'-r_{1}}{d-r_{1}}}\ge\sqrt{1+\varepsilon}\ .
\]
Hence, the lower bound is maintained and at all times the difference
of the deadline to $r_{j}$ grows by a factor of $\sqrt{1+\varepsilon}$
with every deadline (after the first deadline that is bigger than $r_j+1$). Thus, the number
of deadlines that are at least $r_j+1$ is always bounded by $O(\log_{\sqrt{1+\varepsilon}}(T))=O(\log(T)/\varepsilon)$.

In particular, the procedure must terminate. Now consider some $j>1$.
We only need to show that the initialization of $D(j)$ satisfies
again the lower bound, then we can apply the same argument. To this
end, assume inductively it holds for $D(j-1)$ and recall that $D(j)$
is initialized as $\{r_j,r_{j}+1\}\cup(D(j-1)\cap[r_{j},\infty)$. Thus,
any pair of consecutive deadline $d'>d>r_{j}+1$ also appears in the
initial $D(j-1)$. It follows that that
\[
\frac{d'-r_{j}}{d-r_{j}}\ge\frac{d'-r_{j-1}}{d-r_{j-1}}\ge\sqrt{1+\varepsilon}\ .
\]
Thus each job has at most $O(\log(T)/\varepsilon)$ deadlines that are at least $r_{j}+1$, i.e. there is a constant $C>0$ with $|D(j)\cap [r_j+1, \infty)|\leq C\cdot\log(T)/\varepsilon)$ for each job $j$. 
We prove by induction by $j$ that each job has at most $2C\cdot\log(T)/\varepsilon$ deadlines in total. 
There is no deadline for the first job between $r_1$ and $r_1+1$ by construction, so the induction hypothesis holds. So let $j>1$. If $r_j=r_{j-1}$ then $D(j)=D(j-1)$ as we do not create new deadlines for $j$. 
Thus by the induction hypothesis $|D(j)|=|D(j-1)|\leq 2C\cdot\log(T)/\varepsilon$.	
If $r_j>r_{j-1}$ then $r_{j}\geq r_{j-1}+1$ and for $j$ the recursive algorithm only introduces new deadlines that are at least $r_j+1$. So $D(j)\cap (r_j, r_{j}+1)\subseteq D(j-1)\cap (r_{j-1}+1, \infty)$ and $|D(j)\cap (r_j, r_{j}+1)|\leq |D(j-1)\cap (r_{j-1}+1, \infty)|\leq C\cdot\log(T)/\varepsilon)-1 $. 
This yields $|D(j)|=|\{r_j\}|+|D(j)\cap (r_j, r_{j}+1)|+|D(j)\cap [r_j, \infty)|\leq 2C\cdot\log(T)/\varepsilon$. Thus each job has at most $O(\log(T)/\varepsilon)$ deadlines.

Let us now prove the statement of the lemma. Since the number of discrete
deadlines is at most twice the number of non-discrete deadlines, the
first property is satisfied. The second property and fourth properties
follow immediately from construction. For the third property, consider
two consecutive discrete deadlines of job $j$ that are at least $r_j+1$. If both deadlines
were derived from the same non-discrete deadline, i.e., they are $\lfloor d\rfloor_{\delta}$
and $\lceil d\rceil_{\delta}$ for some non-discrete deadline $d$,
then $\lceil d\rceil_{\delta}=\lfloor d\rfloor_{\delta}+\delta$ and
thus $(\lceil d\rceil_{\delta}-r_{j})/(\lfloor d\rfloor_{\delta}-r_{j})\le(1+\delta)/1\le1+\varepsilon$.
Otherwise, the two discrete deadlines must be $\lfloor d'\rfloor_{\delta}$
and $\lceil d\rceil_{\delta}$ for two consecutive non-discrete deadlines
$d'>d$. In that case,
\[
\frac{\lfloor d'\rfloor_{\delta}-r_{j}}{\lceil d\rceil_{\delta}-r_{j}}\le\frac{d'-r_{j}}{d-r_{j}}\le1+\varepsilon.
\]
The running time of the procedure is polynomial in the number of deadlines
created and the number of jobs, thus polynomial in $n$ and $\log(T)$.
% \end{proof}

\subsection{Proof of Lemma~\ref{lem:dp}}
% \begin{lemma}\label{lem:dp} The dynamic program above computes the
% optimal cost schedule where jobs can only be started, preempted, or
% resumed at integral multiples of $\delta$ and the cost of a job $j$
% is $w_{j}\cdot(d-r_{j})^{p}$ for the smallest $d\in D(j)$ with $d\ge C_{j}$.
% \end{lemma}
% \begin{proof}
We argue inductively that for each job $j$ and each load vector $L$,
the dynamic program stores in the DP-cell $(j,L)$ a schedule with
minimum cost among all schedule in which for each interval $I\in\mathcal{I}(j)$
and each set of machines $S\subseteq M$, there are exactly $L(I,S)$
many $\delta$-intervals in $I$ during which the jobs $\{j,j+1,\dotsc,n\}$
are scheduled on exactly the set of machines $S$---assuming there
exists such a schedule.

First assume that $j=n$ and consider a load vector $L$. Let $\OPT(j,L)$
be the optimal solution of the subproblem corresponding to the DP-cell
$(j,L)$. For each interval $I\in\I(j)$ and each set $S\subseteq M$
we define a value $y(I,S,i)$ as follows. We define $y(I,S,i):=0$
if $S\neq\{i\}$ and otherwise we define $y(I,S,i)$ as the number
of $\delta$-intervals in $I$ during which $\OPT(j,L)$ processes
job $j$ on machine $i$, i.e., $y(I,\{i\},i):=L(I,\{i\})$. Then,
Inequality~\eqref{eq:dp-complete} must be satisfied since otherwise
$j$ would not be completed in $\OPT(j,L)$. Further, the residual
load vector $L'$ as defined in~\eqref{eq:dp-residual} must be the
all-zero vector. Hence, if we guess $y(I,S,i)$ for each combination
of an interval $I\in\I(j)$, each set $S\subseteq M$, and each $i\in S$,
we obtain a feasible solution. Therefore, we store a solution in the
DP-cell $(j,L)$ and this solution does not have a higher cost than
$\OPT(j,L)$.

Now consider a job $j<n$ and a load vector $L$. Let again $\OPT(j,L)$
be the optimal solution of the subproblem corresponding to the DP-cell
$(j,L)$. For each interval $I\in\I(j)$, each set $S\subseteq M$,
and each $i\in S$ we define a value $y(I,S,i)$ to be the number
of $\delta$-intervals in $I$ during which $\OPT(j,L)$ processes
job $j$ on machine $i$ and some job from the set $\{j+1,j+2,\dotsc,n\}$
on each machine in $S\setminus\{i\}$. Then, Inequality~\eqref{eq:dp-complete}
is satisfied, since again $j$ would not be completed otherwise. For
each interval $I\in\I(j)$ and each set $S\subseteq M$, let $L'(I,S)$
be the number of $\delta$-intervals in $I$ during which $\OPT(j,L)$
processes on each machine in $S$ some job in $\{j+1,j+2,\dotsc,n\}$.
Then $L'$ must satisfy~\eqref{eq:dp-residual}. We define a load
vector $L''$ similarly as $L'$ but with respect to the intervals
in $\mathcal{I}(j+1)$ instead of $\mathcal{I}(j)$. Formally, for
each interval $I\in\I(j+1)$ and each set $S\subseteq M$ let $L''(I,S)$
be the number of $\delta$-intervals in $I$ during which $\OPT(j,L)$
processes on each machine in $S$ some job in $\{j+1,j+2,\dotsc,n\}$.
Then, $L''$ and $L'$ satisfy Inequality~\eqref{eq:dp-subdivision} and $L''$ is compatible with $L'$.
Thus, if we guess the values $y(I,S,i)$, we obtain a feasible solution.
Hence, we store a solution in the DP-cell $(j,L)$ whose cost is not
larger than the cost of $\OPT(j,L)$.

It is straight-forward that in each DP-cell $(j,L)$, we store a schedule
for the jobs $\{j,...,n\}$ that respects the load constraints given
by $L$. Hence, there is a load vector $L^{*}$ for which the cell
$(1,L^{*})$ stores a schedule for all jobs $\{1,...,n\}$ with cost
at most $\OPT$. Thus, we output a solution with cost at most $\OPT$.
% \end{proof}

\subsection{Proof of Lemma~\ref{lem:running-time}}
The running time of our DP is dominated by the number of DP-cells
and guesses. We have that $|\mathcal{I}(j)|\le O(\log(T)/\varepsilon)$
for each job $j$. Thus, each load vector $L$ of a DP-cell $(j,L)$
has at most $O(2^{m}\log(T)/\varepsilon)$ entries. Since for every
entry there are at most $T/\delta$ possible values, we have at most
$(T/\delta)^{O(2^{m}\log(T)/\varepsilon)}$ possible load vectors
$L$ in total for each job $j$.

When we process a DP-cell $(j,L)$, we guess a values $y(I,S,i)$
for each combination of an interval $I\in\mathcal{I}(j)$, a set $S\subseteq M$,
and a machine $i\in S$. Thus, the number of these values is bounded
by $O(2^{m}m\log(T)/\varepsilon)$ and for each value $y(I,S,i)$
there are at most $T/\delta$ many options. Hence, the total number
of guesses for all values $y(I,S,i)$ bounded by $(T/\delta)^{O(2^{m}m\log(T)/\varepsilon)}$.
Thus, the running time is bounded by $O(n\log n)$ for sorting the
jobs initially and by $n\cdot(T/\delta)^{O(2^{m}m\log(T)/\varepsilon)}$
for the dynamic program.

\subsection{Proof of Theorem~\ref{thm:QPTAS}}
By Lemma~\ref{lem:chooseDelta}, there is a $(1+\varepsilon)$-approximate
discrete schedule with $\delta=\varepsilon/n^{m}$. Let $C_{j}^{*}$
be the corresponding completion times and $d_{j}^{*}\in D(j)$ minimal
with $C_{j}^{*}\le d_{j}^{*}$. Then $d_{j}^{*}-r_{j}\le(1+\varepsilon)(C_{j}^{*}-r_{j})$.
Hence,
\[
(\sum_{j}w_{j}(d_{j}^{*}-r_{j})^{p})^{1/p}\le(1+\varepsilon)(\sum_{j}w_{j}(C_{j}^{*}-r_{j})^{p})^{1/p}\le(1+\varepsilon)^{2}\OPT.
\]
by Lemma~\ref{lem:dp} the dynamic program computes a schedule of
at most this cost. Rescaling $\varepsilon$ by a constant improves
this to $1+\varepsilon$. For the running time note that by standard
arguments we can bound $r_{j}$ by $n\cdot p_{\max}$ for all $j$:
otherwise the instance can be split easily into two independent parts.
This means that $T$ is polynomially bounded in $n$ and $p_{\max}$
and the running time follows from Lemma~\ref{lem:running-time}.

\end{document}